\documentclass[aps,pra,reprint,twocolumn,superscriptaddress,nofootinbib]{revtex4-1}
\pdfoutput=1 

\usepackage{graphicx}
\usepackage{bm}
\usepackage{amsmath}
\usepackage{amssymb}
\usepackage{amsfonts}
\usepackage{amsthm}
\usepackage{mathtools}
\usepackage{hyperref}
\usepackage{braket}
\usepackage[normalem]{ulem}
\usepackage{wrapfig}
\usepackage{tikz}
\usepackage{dsfont}
\usepackage{comment}
\usepackage{thmtools,thm-restate}
\mathtoolsset{showonlyrefs}
\usepackage[export]{adjustbox}

\usepackage{soul} 

\hypersetup{colorlinks=true,urlcolor=[rgb]{0,0,0.5},citecolor=[rgb]{0.5,0,0},linkcolor=[rgb]{0,0,0.4}}

\usetikzlibrary{decorations.pathmorphing}

\newtheorem{theorem}{Theorem}
\newtheorem{definition}{Definition}
\newtheorem{corollary}{Corollary}
\newtheorem{lemma}{Lemma}
\newtheorem{proposition}{Proposition}
\newtheorem*{theorem*}{Theorem}

\newtheorem*{task*}{Task}
\newtheorem*{proposition*}{Proposition}

\def\autorefapp#1{\hyperref[#1]{Appendix~\ref{#1}}}

\def\and{\quad {\rm and} \quad}

\newcommand{\calD}{\mathcal{D}}

\newcommand{\calS}{\mathcal{S}}

\DeclarePairedDelimiter{\brk}{[}{]}

\makeatletter
\def\Pr{\@ifnextchar[{\@witha}{\@withouta}}
\def\@witha[#1]{\mathop{\operator@font Pr}_{#1}\brk}
\def\@withouta{\mathop{\operator@font Pr}\brk}
\makeatother

\makeatletter
\def\E{\@ifnextchar[{\@withb}{\@withoutb}}
\def\@withb[#1]{\mathop{\mathbb{E}}_{#1}\brk}
\def\@withoutb{\mathop{\mathbb{E}}\brk}
\makeatother

\makeatletter
\def\Var{\@ifnextchar[{\@withc}{\@withoutc}}
\def\@withc[#1]{\mathop{\mathbb{V}}_{#1}\brk}
\def\@withoutc{\mathop{\mathbb{V}}\brk}
\makeatother

\renewcommand{\vec}[1]{\mathbf{#1}}

\newtheorem*{CM3p2}{Theorem 3.2 of~\cite{collins2017weingarten}}

\newtheorem*{ACQl6}{Lemma 6 of~\cite{aharonov2021quantum}}

\newtheorem*{rephrase}{Theorem 2$\,'$}

\begin{document}
\title{Information-theoretic Hardness of Out-of-time-order Correlators}

\author{Jordan Cotler}
\email{jcotler@fas.harvard.edu}
\affiliation{Society of Fellows, Harvard University, Cambridge, MA, USA}
\affiliation{Black Hole Initiative, Harvard University, Cambridge, MA, USA}

\author{Thomas Schuster}
\affiliation{Google Quantum AI, 340 Main Street, Venice, CA 90291, USA}
\affiliation{Department of Physics, University of California, Berkeley, California 94720 USA}

\author{Masoud Mohseni}
\affiliation{Google Quantum AI, 340 Main Street, Venice, CA 90291, USA}

\begin{abstract}
We establish that there are properties of quantum many-body dynamics which are efficiently learnable if we are given access to out-of-time-order correlators (OTOCs), but which require exponentially many operations in the system size if we can only measure time-ordered correlators.  This implies that any experimental protocol which reconstructs OTOCs solely from time-ordered correlators must be, in certain cases, exponentially inefficient.  Our proofs leverage and generalize recent techniques in quantum learning theory.  Along the way, we elucidate a general definition of time-ordered versus out-of-time-order experimental measurement protocols, which can be considered as classes of adaptive quantum learning algorithms.  Moreover, our results provide a theoretical foundation for novel applications of OTOCs in quantum simulations.
\end{abstract}

\maketitle

\section{Introduction}

\vspace{-.1cm}

There has been a surge of interest in quantum many-body chaos in recent years, in part due to the popularization of computable probes such as out-of-time-order correlators (OTOCs)~\cite{larkin1969quasiclassical, kitaev2014hidden, maldacena2016bound}. These correlators capture the quantum analog of the butterfly effect, wherein small changes to initial conditions lead to exponentially large changes at later times.  This theoretical tool motivated an experimental interest in measuring the quantum butterfly effect, leading to a myriad of proposals for implementations~\cite{swingle2016measuring,yao2016interferometric,yoshida2019disentangling,vermersch2019probing,qi2019measuring} which have been actualized in progressively more sophisticated nuclear magnetic resonance (NMR) setups~\cite{baum1985multiple,li2017measuring,sanchez2021emergent,dominguez2021decoherence}, quantum many-body simulators~\cite{garttner2017measuring,joshi2020quantum}, and quantum computers~\cite{landsman2019verified,blok2020quantum,mi2021information}.

A notable feature of OTOCs, embedded in their name, is that they involve correlations within a system that are extracted by evolving the system both forwards and backwards in time in an alternating fashion.  This is unlike more conventional experimental correlators in which probes are only introduced as time progresses forward. While such time reversal mid-experiment is possible for certain experimental setups~\cite{baum1985multiple,garttner2017measuring,mi2021information}, it is impractical for most systems.  Indeed, in response to this difficulty, many proposed experimental protocols for measuring OTOCs do not entail reversing time directly, but rather reconstruct the effects of time reversal from experiments having only forwards time evolution~\cite{yao2016interferometric,vermersch2019probing,qi2019measuring,joshi2020quantum}.

It is interesting to ask if the ability to reverse time mid-experiment could be a valuable tool for learning properties of systems in nature.  Indeed, the works mentioned above appear to suggest that the quantum butterfly effect necessitates the reversal of time in order to be measured efficiently.  More generally, it seems plausible that the ability to reverse time could enable certain properties of an experimental system to be revealed with greater efficiency.   A recent work by the authors explored this idea with extensive numerical examples, which indicated affirmatively that time reversal can provide potential gains for certain learning problems~\cite{schuster2022learning}.  If we further consider systems in nature which cannot be directly time-reversed but which have well-characterized dynamics, then a quantum simulation incorporating time-reversal could still be performed.  In this context as well, time-reversal could unveil otherwise difficult-to-access properties of the simulated system.  

In the present work, we rigorously establish that there are properties of quantum many-body systems which are efficiently learnable with OTOCs, but which require exponentally many operations in the system size if we can only measure time-ordered correlators.  An interesting corollary is that any experimental protocol which reconstructs OTOCs from time-ordered measurements must be, in certain cases, exponentially inefficient.  Our proofs build upon and generalize recent work on quantum algorithmic measurements (QUALMs) and quantum learning theory~\cite{aharonov2021quantum, chen2021exponential, chen2021hierarchy, huang2021quantum}.

The formulation of our results entails a precise definition of the most general time-ordered and out-of-time-order correlators that can be measured; this may be of interest for other applications.  Our definition considers learning protocols for measuring properties of a physical system and its time evolution, leveraging the learning tree framework of~\cite{chen2021exponential}.

The remainder of the paper is organized as follows.  In Section~\ref{sec:review} we review OTOCs and their role in quantum many-body chaos.  In Section~\ref{sec:defs} we formulate the most general time-ordered experiments, and the most general out-of-time-order experiments.  In Section~\ref{sec:result} we explain our main results on the hardness of measuring certain properties of quantum many-body systems using time-ordered operations alone.  We conclude in Section~\ref{sec:discussion} with a discussion.

\vspace{-.5cm}

\section{Review of OTOCs}
\label{sec:review}

\vspace{-.1cm}

Here we briefly review OTOCs in the context of quantum chaos. We begin with the setting of classical mechanics, and for concreteness consider a phase space $(x_1,x_2,x_3, p_1, p_2, p_3)$ equipped with Hamiltonian dynamics.  If the system starts at position $x_0 = (x_1(0),x_2(0),x_3(0))$ and is evolved to time by a time $t$, then we denote its new position by $x(t,x_0)$, although we will often suppress the dependence on the initial condition.  To ascertain the sensitivity to initial conditions, we can compute the derivative of the first coordinate $x_1(t)$ with respect to, say, the first coordinate of the initial condition $x_1(0)$, giving $\frac{\partial x_1(t)}{\partial x_1(0)}$.  For chaotic systems this quantity can exhibit exponential growth in $t$, with a growth rate characterized by a so-called Lyapunov exponent.  This is the classical butterfly effect: the behavior of the system at later times is exponentially sensitive to the choice of initial conditions.

To motivate a quantum generalization, we can write $\frac{\partial x_1(t)}{\partial x_1(0)} = \{x_1(t), p_1(0)\}_{\text{PB}}$ where the right-hand side is the Poisson bracket.  The quantum version of this quantity is naturally $\frac{1}{i\hbar}[\hat{x}_1(t), \hat{p}_1(0)]$, as was suggested in the seminal work of Larkin and Ovchinnikov~\cite{larkin1969quasiclassical}.  Notice that this object is an operator; since it is convenient to have a single number which captures the exponential growth of chaos, it is natural to take the expectation value of the operator with respect to a state $\rho$ as $\frac{1}{i\hbar}\,\text{tr}(\rho\,[\hat{x}_1(t), \hat{p}_1(0)])$.  Often $\rho$ is chosen to be a thermal state, although in this case the expectation value may fluctuate around zero.  To ameliorate this, the expectation value of the \textit{square} of the commutator can be considered, namely $- \frac{1}{\hbar^2}\,\text{tr}(\rho \, [\hat{x}_1(t), \hat{p}_1(0)]^2)$.  Expanding this out, there are terms of the form $- \frac{1}{\hbar^2}\,\text{tr}(\rho\,\hat{x}_1(t)\,\hat{p}_1(0)\,\hat{x}_1(t)\,\hat{p}_1(0))$, which are indeed out-of-time-order: we start at time zero, evolve to time $t$, evolve back to time zero, and evolve back to time $t$.  It is this OTOC term that gives the quantum analog of exponential growth.

More generally, in quantum many-body systems the preferred version of the OTOC is often $\text{tr}(\rho\, [W(t), V(0)]^2)$ where $\rho$ is a thermal state and $W(0), V(0)$ are (initially) spatially local operators~\cite{kitaev2014hidden, maldacena2016bound}.  Then the OTOC measures how much $W(t)$ and $V(0)$ fail to commute in the Heisenberg picture, which for certain systems can grow exponentially; this growth is contained in the out-of-time-order term $\text{tr}(\rho\, W(t) V(0) W(t) V(0))$.  Such correlators have been extensively studied and characterized (e.g.~\cite{shenker2014black, stanford2016many, maldacena2016bound, maldacena2016remarks, hosur2016chaos, cotler2017chaos, cotler2018out, nahum2018operator, von2018operator}). 

While correlators are natural objects in quantum systems and field theories, they are often studied abstractly without acknowledgement of how they might be measured in a physical system.  The question of devising a measurement protocol to extract a particular correlator from a system of interest is particularly pressing in the case of OTOCs.  Below we will use the framework of learning theory to provide a general definition of how correlators can be obtained via quantum measurements.

\vspace{-.4cm}

\section{Time-ordered and Out-of-time-order experiments}
\label{sec:defs}

\vspace{-.1cm}

In this section we mathematically formalize how correlators are extracted from measurements of a system.  We leverage the learning tree formalism for quantum channels, developed in~\cite{aharonov2021quantum, chen2021exponential}.  Let us outline an intuitive understanding of how such experiments operate, and then render this into more precise definitions.

Suppose we have some experimental system with time evolution by a unitary $U$, which is not known to or fully characterized by the experimentalist.  The experimentalist desires to learn about $U$ by making measurements on the system as it evolves.  For a Hamiltonian system, we might have $U = e^{- i H\, \Delta t}$ for some not fully characterized $H$, where $\Delta t$ is the shortest time scale over which we can control the evolution.  So if the experimentalist wants to evolve the system by a time $k \,\Delta t$, he can simply apply $U^k$.  (Our formulation will also work if the experimentalist has continuous control over the time $t$, but this $\Delta t$ discretization will make our definitions simpler to state.)  To be explicit, we stipulate that the system in question is composed of $n$ qubits on which the unitary $U$ acts.

An experiment would operate as follows.  The experimentalist begins by preparing the system in some initial state $\rho_0$.   Thereafter, he can choose to either: (i) apply $U$;  (ii) apply some other quantum channel; or (iii) perform a partial or complete measurement, which would confer some classical information about the state of the system which he could store on a classical computer.  He can exercise these options again and again in a sequence, each time basing his decision of what to do next on the information collected thus far.  That is, the protocol for information collection can be \textit{adaptive}.  At the end of the experiment, the classical computer contains the information the experimentalist has gained by performing measurements at any stage throughout the protocol.

We will make the assumption that throughout the protocol, the state of the system is not entangled with any external ancilla system which the experimentalist can manipulate.  This choice is made to reflect contemporary experimental realities; for instance, at present, there is no way of entangling a sample of graphene to an external quantum computer.  In the Discussion, we will comment further on the possibility of ancilla-assisted protocols.

Next we turn to formalizing the notion of an experiment explained above.  First we note that in our setting, the most general operation the experimentalist can perform on a quantum state is a POVM measurement~\cite{nielsen2002quantum}.  That is, consider a collection of operators $\{F_i\}_i$ on $n$ qubits, satisfying $\sum_i F_i^\dagger F_i = \mathds{1}$.  Then the POVM measurement with respect to these operators maps
\begin{equation}
\rho \,\,\longmapsto\,\, \frac{F_i \rho F_i^\dagger}{\text{tr}(F_i \rho F_i^\dagger)}\quad \text{with probability} \,\,\, \text{tr}(F_i \rho F_i^\dagger)\,.
\end{equation}
Here the experimentalists' apparatuses would register that he had measured the $i$th outcome.  This generalizes the notion of a projective measurement from elementary quantum mechanics.  Note that a special case of a POVM measurement is simply the application of a unitary $V$; if the POVM is the singleton set $\{V\}$ which clearly satisfies $V^\dagger V = \mathds{1}$, then `measuring' $\rho$ yields $V \rho V^\dagger$ with probability one.

The reason that POVM measurements are so general is encapsulated in the following fact: any composition of quantum channels and POVM measurements can be captured by a \textit{single} new POVM measurement.  That is, suppose we have a quantum state and suscept it to a sequence of POVM measurements and quantum channels; then the result of this is the same as having applied some \"{u}ber-POVM measurement.

With the above in mind, we can conceive of an experimental protocol as occurring in a sequence of rounds.  The protocol is as follows:
\begin{itemize}
    \item Initialize $\rho_0$.
    \item Apply $U$, measure the state using a POVM $\{F_i\}_i$.  Suppose the outcome is $i = q$\,; then store this in the classical memory.  The output is the state $\rho_q := \frac{F_q U \rho_0 U^\dagger F_q^\dagger}{\text{tr}(F_q U \rho_0 U^\dagger F_q^\dagger)}$.
    \item Apply $U$, measure the state using a POVM $\{F_{q,i}\}_i$, which can be contingent on the previous measurement outcome $q$.  Suppose the new outcome is $i = r$\,; then store this in the classical memory.  The output is the state $\rho_{q,r} := \frac{F_{q,r} U \rho_q U^\dagger F_{q,r}^\dagger}{\text{tr}(F_{q,r} U \rho_q U^\dagger F_{q,r}^\dagger)}$.
    \item Apply $U$, measure the state using a POVM $\{F_{q,r,i}\}_i$, which can be contingent on the previous measurement outcomes $q,r$.  Suppose the new outcome is $i = s$\,; then store this in the classical memory.  The output is the state $\rho_{q,r,s} := \frac{F_{q,r,s} U \rho_{q,r} U^\dagger F_{q,r,s}^\dagger}{\text{tr}(F_{q,r,s} U \rho_{q,r} U^\dagger F_{q,r,s}^\dagger)}$.
    \item Repeat this kind of adaptive POVM measurement procedure for $T$ total rounds.
\end{itemize}
\begin{figure}[t!]
    \centering
    \includegraphics[width=0.4\textwidth]{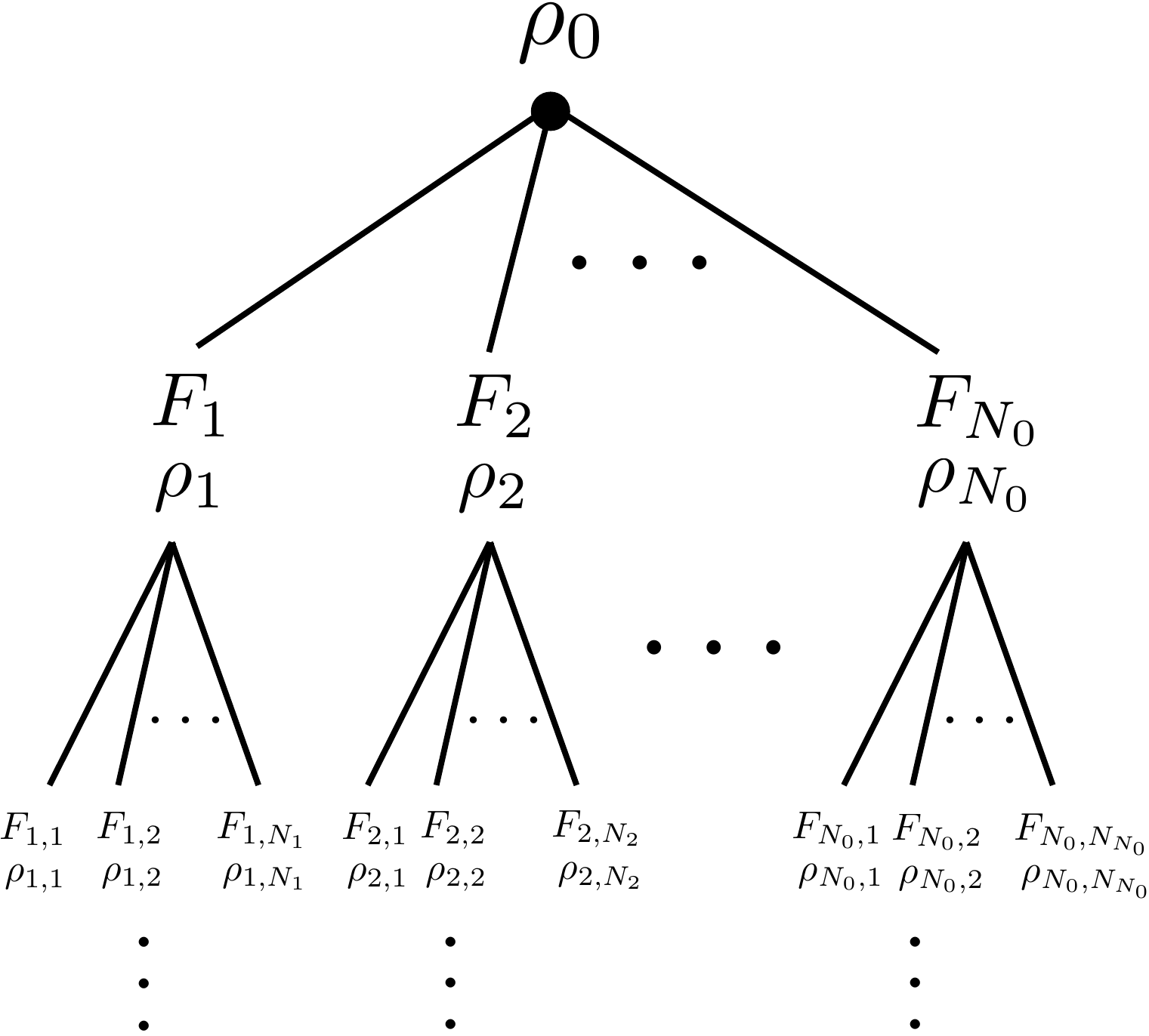}
    \centering
    \caption{\textit{Depiction of a learning tree $\mathcal{T}$.} We begin with a state $\rho_0$ at the root of the tree, and perform successive, adaptive POVM measurements.  A root-to-leaf path through the tree corresponds to a sequence of POVM measurement outcomes.}
    \label{fig:Tree1}
\end{figure}
Here, the (adaptive) sequence of POVMs only `knows' about $U$ via the sequence of measurement outcomes.

Several further comments are in order.  First, note that a POVM measurement is not applied between the initial preparation of $\rho_0$ and the initial application $U$; this would be superfluous since it is equivalent to having prepared a different initial state.  Second, this protocol is clearly adaptive, since the choice of each POVM measurement can be contingent on all previous measurement outcomes.  Indeed, the learning tree specifies an adaptive strategy since it prescribes how the experimentalist makes his adaptive choices.  Third, observe that if the experimentalist wanted to apply $k$ $U$'s in a row, i.e.~$U^k$, he could simply choose for the POVMs to be $\{\mathds{1}\}$ for $k$ rounds in a row.  Finally, we observe that the protocol outlined above is only directly capturing \textit{time-ordered} correlations, since the experimentalist can apply $U$ but not $U^\dagger$.

This class of protocols can be fruitfully organized into a tree, as per Figure~\ref{fig:Tree1}.  We start at the root of the tree (i.e.~the top-most vertex), and traverse down the tree by successively performing POVM measurements in an adaptive fashion.  We see, then, that a particular instantiation of the protocol is a root-to-leaf (i.e.~top-to-bottom) path through the tree.  A tree of depth $T$ corresponds to applying $U$ a total of $T$ times, i.e.~once per round.  The classical information that the experimentalist obtains is the sequence of POVM measurement outcomes, which corresponds to a root-to-leaf path through the tree. A path is labelled by a sequence of vertices $v_0, v_1, ..., v_T = \ell$, or more simply by $\ell$ since the leaf node specifies the entire root-to-leaf path.

Let us denote such a learning tree by $\mathcal{T}$.  It represents a specification of an adaptive experimental protocol that an experimentalist can perform. We provide more formal details in the Appendix.  Now if $v_0, v_1, ..., v_T = \ell$ is a root-to-leaf path through $\mathcal{T}$, then the probability of taking that path is
\begin{equation}
p^U(\{v_t\}) := \prod_{t=1}^T \text{tr}(F_{v_t} U \rho_{v_{t-1}} U^\dagger F_{v_t}^\dagger)\,,
\end{equation}
which can be more conveniently notated by $p^U(\ell)$.  In other words, this is the probability of the experimentalist obtaining the sequence of measurement outcomes given by the root-to-leaf path through the tree terminating in $\ell$.  The way that information is extracted from an experiment is via a function $G(\ell)$ which maps the sequence of measurement outcomes to the value of some desired quantity, e.g.~a time-ordered correlator.  The empirical expectation value of $G(\ell)$ is then
$\hat{G} := \mathbb{E}_{p^U(\ell)}[G(\ell)] = \sum_\ell p^U(\ell) \, G(\ell)$.

The above motivates the following definition of a time-ordered experiment for learning properties of $U$, which we further detail in the Appendix:
\begin{definition}[Time-ordered experiment]
A time-ordered experiment is any learning tree protocol $\mathcal{T}$ which queries $U$.
\end{definition}
\noindent The definition of an out-of-time-order experiment follows in a similar fashion:
\begin{definition}[Out-of-time-order experiment]
An out-of-time-order experiment is any learning tree protocol $\mathcal{T}'$ which queries both $U$ and $U^\dagger$, where the choice of which one is to be queried in each round can be determined adaptively.
\end{definition}
\noindent Now suppose we want to measure an OTOC such as $\text{tr}(\rho_0 U^\dagger W U V U^\dagger W U V )$.  Clearly this is most accessible with an out-of-time-order experiment.  However, we emphasize that we can obtain this OTOC using the data of a time-ordered experiment, although we might require many more rounds of the experiment to obtain the answer to within the desired precision.

Indeed, our goal in next section is to establish that if we do not fully know $U$ (or $U^\dagger$), then there are certain OTOCs which are readily and efficiently attained by an out-of-time-order experiment, but which require exponentially many operations if the experiments are time-ordered.

\vspace{-.2cm}

\section{Information-theoretic hardness of OTOCs}
\label{sec:result}

\vspace{-.1cm}

In this section we explain our main result, namely that for quantum many-body systems with partially unknown dynamics, there can be OTOCs which are easy to measure with out-of-time-order experiments but which are exponentially hard to measure with only time-ordered experiments.  Said differently, any experimental protocol that reconstructs OTOCs from only time-ordered experiments must in some cases be exponentially inefficient.  In this manner, our results elucidate fundamental differences between OTOCs and time-ordered correlators.

Our proof strategy is to construct an explicit example for which measuring an OTOC to within constant error has an exponential disparity between the time-ordered and out-of-time-ordered experimental settings. Concretely, consider again an $n$ qubit system, here for $n$ even, equipped with a partially uncharacterized unitary.  Suppose that it is either: (i) a fixed, Haar-random unitary $U$ on $n$ qubits, or (ii) a product $U_1 \otimes U_2$ of two fixed, Haar-random unitaries $U_1, U_2$, each on $n/2$ qubits.  Here $U_1$ is to act on the first $n/2$ qubits, and $U_2$ is to act on the remaining $n/2$ qubits.  The experimentalist will not know which of these two possibilities (i) or (ii) is the case, and is tasked with performing an experimental protocol to determine which one is instantiated.

The two possibilities are physically rather different.  In (i) all of the qubits interact with one another, whereas in (ii) only blocks of half of the qubits mutually interact.  This suggests that if the experimentalist can perform an out-of-time-order experiment, it is quite easy to distinguish between (i) and (ii) by measuring a single OTOC.  This works in the following way.  The experimentalist prepares the system in the all zero state $|0\rangle^{\otimes n}$, and then applies the unknown unitary.  Thereafter, the experimentalist applies $\sigma_x$ on the first qubit to flip it, followed by applying the inverse of the unknown unitary.  Then the experimentalist checks if the second block of $n/2$ qubits is again in the all zero state.  This corresponds to measuring the OTOC
\begin{equation}
\textsf{OTOC}(V) = \text{tr}\!\left(\mathds{1}_{\frac{n}{2}} \! \otimes |0\rangle \langle 0|^{\otimes \frac{n}{2}}\!\left\{V^\dagger \sigma_x^1 V |0\rangle \langle 0|^{\otimes n} V^\dagger \sigma_x^1 V\right\}\right),
\end{equation}
where $V$ is a placeholder for the unknown unitary.  In case (i), the final output state will be complicated, having little overlap with the all zero state.  Indeed, on average we have
\begin{equation}
\label{E:expectbound1}
\mathbb{E}_{U \sim \text{Haar}(2^n)}\!\left[\textsf{OTOC}(U)\right] = \frac{2^{\frac{3n}{2}} - 1}{2^{2n} - 1} \leq O(1/2^{n/2})\,.
\end{equation}
However, in case (ii) the $\sigma_x^1$ operator still allows $U_2$ to cancel with $U_2^\dagger$, and so the second block of $n/2$ qubits ends up precisely in the all zero state.  In terms of the OTOC correlator, we have
\begin{equation}
\label{E:noexpectbound1}
\textsf{OTOC}(U_1 \otimes U_2) = 1 \quad \text{for all }\,\,U_1, U_2\,.
\end{equation}
These results are illustrated in Figure~\ref{fig:UnitaryFig1}.

\begin{figure}[t!]
    \centering
    \includegraphics[width=0.45\textwidth]{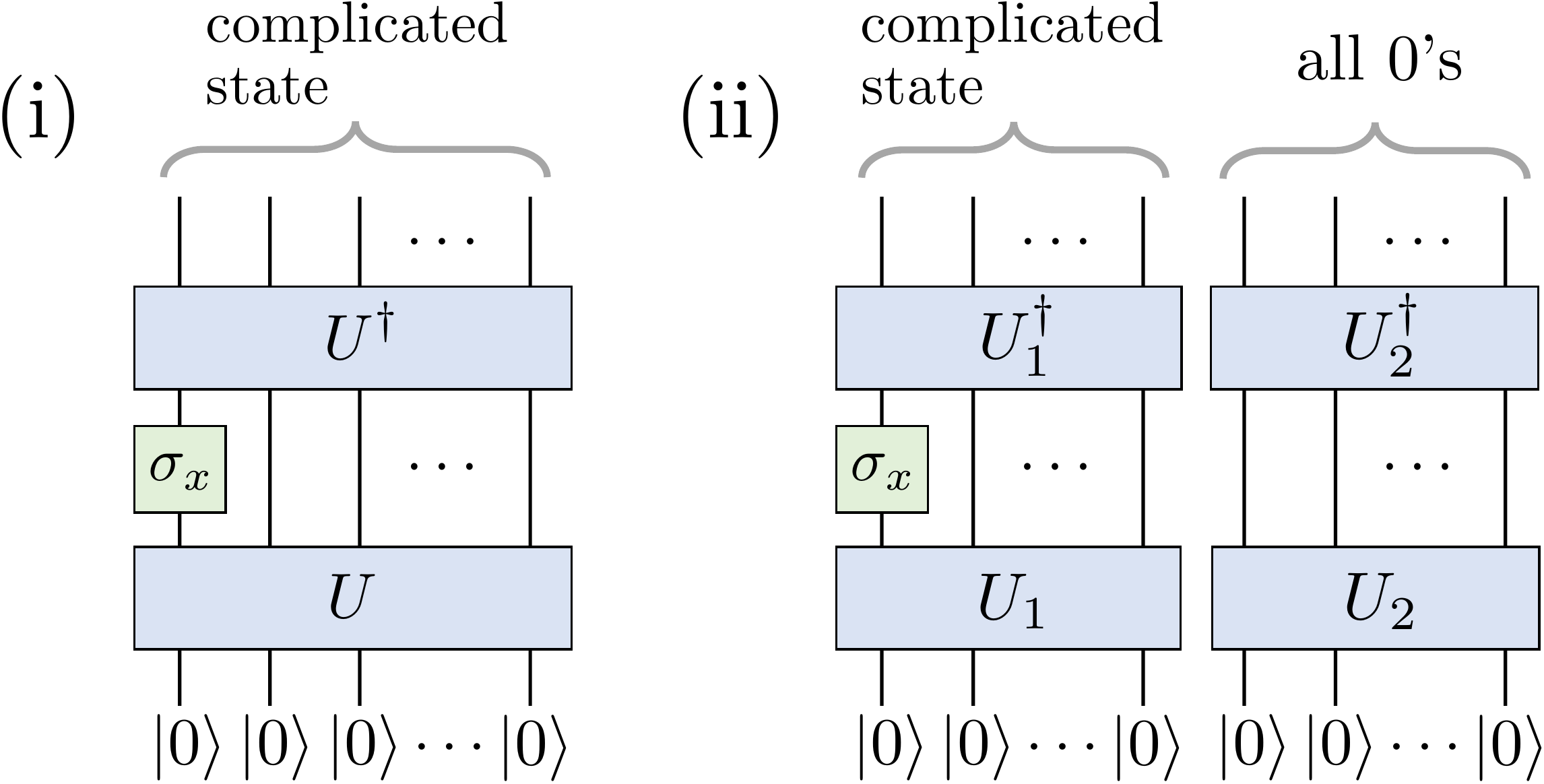}
    \centering
    \caption{\textit{Schematic of out-of-time-order experiment.}  (i) In the first case, applying $U$, then $\sigma_x^1$, and then $U^\dagger$ results in a complicated state. (ii) In the second case, applying $U_1 \otimes U_2$, then $\sigma_x^1$, and then $U_1^\dagger \otimes U_2^\dagger$ leads to a product of two pure states, each on $n/2$ qubits; the first is complicated, the second is the all zero state.}
    \label{fig:UnitaryFig1}
\end{figure}

More formally, these results have the following consequence:
\begin{theorem}[Easiness of task with out-of-time-order experiment] If the experimentalist can perform an out-of-time-order experiment, then with probability exponentially close to one the cases $\text{\rm (i)}$ and $\text{\rm (ii)}$ can be distinguished using only a single application of the unknown unitary, and a single application of its inverse.
\end{theorem}
\begin{proof}
In case (i), the probability that the OTOC is less than or equal to a small constant $\varepsilon$ is bounded by Markov's inequality, namely
\begin{align}
\text{Prob}[\textsf{OTOC}(U) \leq \varepsilon] &\geq 1 - \frac{\mathbb{E}_{U \sim \text{Haar}(2^n)}\!\left[\textsf{OTOC}(U)\right]}{\varepsilon} \nonumber \\
&\geq 1 - O(1/(\varepsilon\,2^{\frac{n}{2}}))\,,
\end{align}
where we have used~\eqref{E:expectbound1}.  In case (ii), the probability that $\textsf{OTOC}(U_1 \otimes U_2)$ is greater than $\varepsilon$ is one, on account of~\eqref{E:noexpectbound1}.  Thus the two possibilities can be distinguished with probability exponentially close to unity, and the protocol only requires a single query of the unknown unitary and a single query of its inverse.
\end{proof}
This result lies in contrast to the following, more difficult theorem:
\begin{theorem}[Exponential hardness of task with time-ordered experiment]
Any time-ordered experiment which can distinguish cases $\text{\rm (i)}$ and $\text{\rm (ii)}$ requires at least $\Omega(2^{n/4})$ queries of the unknown unitary, and so is exponentially inefficient.
\end{theorem}
\noindent An equivalent rephrasing is our promised result:
\begin{rephrase}[Exponential hardness OTOCs with time-ordered experiments]
Any time-ordered experimental protocol to determine an OTOC to within constant error must in certain instances require accessing the time evolution exponentially many times.  However, in some of these instances, an out-of-time-order experiment can determine the OTOC to within constant error by accessing the time evolution only a constant number of times.
\end{rephrase}
\noindent The theorems should be regarded as quantifying a form of information-theoretic hardness, since they bound the number of applications of the unknown unitary that we need to succeed in the time-ordered setting.  
While the proof of hardness is given in the Appendix, we sketch its high-level strategy here.

The idea, coming from previous work~\cite{aharonov2021quantum, chen2021exponential}, is to upper bound the sum
\begin{equation}
\label{E:tobound1}
\sum_{\ell \,\in\, \text{leaf}(\mathcal{T})}\left|\mathop{\mathbb{E}}_{U \sim \text{Haar}(2^n)}[p^U(\ell)] - \mathop{\mathbb{E}}_{U_1, U_2 \sim \text{Haar}(2^{n/2})}[p^{U_1 \otimes U_2}(\ell)]\right|
\end{equation}
for any time-ordered experiment $\mathcal{T}$.  We would like to show that this quantity is $o(1)$ if the number of applications of the unknown unitary is less than $o(2^{n/4})$.    This would imply that the probability of distinguishing (i) and (ii) can only reach a constant value (i.e.~one that is not suppressed in $n$) if we apply the oracle exponentially many times in $n$.  Intuitively, upper bounding~\eqref{E:tobound1} by a small number means that the probability distribution over measurement outcomes looks extremely similar regardless of whether case (i) or case (ii) is instantiated; this means that the two cases cannot be distinguished.

Operationally, we show that $\mathop{\mathbb{E}}_{U \sim \text{Haar}(2^n)}[p^U(\ell)]$ and $\mathop{\mathbb{E}}_{U_1, U_2 \sim \text{Haar}(2^{n/2})}[p^{U_1 \otimes U_2}(\ell)]$ are each close to the uniform distribution $\frac{1}{2^{nT}}$, and hence close to one another via the triangle inequality.  To establish closeness to the uniform distribution, we rewrite the Haar averages in terms of sums of correlators involving permutation operators via the Weingarten calculus.  In each case, one of the correlators is exponentially close to $\frac{1}{2^{nT}}$, and intricate algebraic manipulations establish that the remaining correlators are further suppressed by factors of the Hilbert space dimension. The required technical tools including Haar integration, Weingarten functions, and the learning tree formalism are provided in the Appendix A. Our main proofs are presented in Appendix B.

\vspace{-.2cm}

\section{Discussion}
\label{sec:discussion}

\vspace{-.1cm}

In this paper we have given a precise framework for defining and analyzing time-ordered versus out-of-time-order experiments, and established that the latter have an exponential advantage over the former for measuring certain OTOCs.  Our methods advance recent developments in quantum learning theory~\cite{aharonov2021quantum, huang2021information, chen2021exponential, chen2021hierarchy, huang2021quantum}, and are also a testament to the power of its perspective.

It would be interesting to generalize our results to more realistic settings, e.g.~when the unitaries in question are not constructed from Haar-random ensembles.  There has been progress in this vein for a related class of learning problems involving states instead of unitaries, e.g.~\cite{chen2021exponential, huang2021quantum}. 

We emphasize that in this paper we have made the physically reasonable assumption that we cannot entangle our system to ancillas which could act as a quantum memory.  Indeed, in certain cases adding ancillas could equalize the power balance between time-ordered and out-of-time-order experiments for certain OTOC learning tasks.
This tradeoff and tension between quantum memory and the inability to reverse time is worthy of further investigation.

There has been previous work on the difficulty of simulating the Hermitian conjugate of a unitary $U$ given only black box access to $U$~\cite{quintino2019probabilistic, quintino2019reversing}; our approach in the present work is different, since we instead consider experiments for learning \textit{properties} of $U$.  However, our results and techniques may interface in interesting ways with this line of previous work, for instance establishing new hardness results.  We note that our Theorem 2 implies that in the worst case it is exponentially hard to construct the inverse of a unknown unitary $U$ for which one has query access; this is consistent with~\cite{quintino2019probabilistic, quintino2019reversing}.

As a conceptual coda to our results, we remark that in our own universe we do not have the ability to reverse the direction of time.  As such, there may be physically interesting features of nature, such as ones pertaining to quantum chaos, which are effectively inaccessible to us.  This is also true of experimental systems in which we cannot, in practice, reverse the direction of their time evolution.  In the latter case, we may one day be able to exercise the option of simulating that physical system on a quantum computer, and so time-reversal becomes available.  Thus the ability to control the flow of time evolution in a quantum computer may ultimately allow us to unlock hidden properties of natural systems around us.
\\ \\ \\
\noindent {\bf Acknowledgments.}\quad 
We thank Hsin-Yuan Huang for valuable discussions, and Jarrod McClean for comments on a draft of this manuscript. JC is supported by a Junior Fellowship from the Harvard Society of Fellows, the Black Hole Initiative, as well as in part by the Department of Energy under grant {DE}-{SC0007870}.  TS acknowledges support from the National Science Foundation Graduate Research Fellowship
Program under Grant No.~DGE 1752814.
$$$$

\pagebreak
\onecolumngrid
\appendix

\section{Technical Preliminaries}

Before delving into the proof of Theorem 1, we require some technical definitions and tools.  Included are more precise versions of Definition 1 and Definition 2 from the main text. 

\subsection{Notation}

We will work with an $n$ qubit Hilbert space $\mathcal{H} \simeq (\mathbb{C}^2)^{\otimes n}$, for $n$ even.  The dimension of the Hilbert space is $2^n$, which we denote by $d$.  Our main proof will involve extensive use of diagrammatic tensor network notation, reviewed in detail in~\cite{chen2021exponential}.  Our conventions for the diagrams will match those of~\cite{chen2021exponential}.

\subsection{Haar integration and Weingarten functions}

Consider the unitary group $U(d)$.  It will be convenient to use multi-index notation, wherein $I = (i_1, ..., i_k)$ and analogously for $I',J,J'$.  Letting
\begin{equation}
U_{IJ}^{\otimes k} := U_{i_1 j_1} U_{i_2 j_2} \cdots U_{i_k j_k}\,,
\end{equation}
we will be interested in computing expectation values of the form $\mathop{\mathbb{E}}_{U \sim \text{Haar}(d)}[U_{IJ}^{\otimes k} U_{J' I'}^{\dagger \, \otimes k}]$.  To write out the result of this expectation value, we denote by $S_k$ the symmetric group on $k$ elements; for $\sigma \in S_k$, we adopt the notation
\begin{equation}
\delta_{\sigma(I), I'} := \delta_{i_{\sigma(1)},i_1'} \delta_{i_{\sigma(2)},i_2'} \cdots \delta_{i_{\sigma(k)},i_k'}\,.
\end{equation}
Then we have the useful identity (see e.g.~\cite{kostenberger2021weingarten} for a review)
\begin{equation}
\mathop{\mathbb{E}}_{U \sim \text{Haar}(d)}[U_{IJ}^{\otimes k} U_{J'I'}^{\otimes k}] = \sum_{\sigma, \tau \in S_k} \delta_{\sigma(I), I'} \delta_{\tau(J), J'} \, \text{Wg}^U(\sigma \tau^{-1},d)\,,
\end{equation}
where $\text{Wg}^U(\,\cdot\,, d) : S_k \to \mathbb{R}$ is the Weingarten function.  This function can be constructed somewhat explicitly in the following way.  In a slight abuse of notation, let us also denote by $\sigma, \tau$ their representation on $\mathcal{H}^{\otimes k}$, and define
\begin{equation}
G^U(\sigma \tau^{-1},d) := \text{tr}(\sigma \tau^{-1}) = d^{\#(\sigma \tau^{-1})}
\end{equation}
where $\#(\sigma \tau^{-1})$ is the number of cycles of $\sigma \tau^{-1}$.  Viewing $G^U(\sigma \tau^{-1},d)$ as a $k! \times k!$ matrix $G_{\sigma^{-1},\tau}$, we have that $\text{Wg}^U(\sigma^{-1} \tau, d)$ is its matrix inverse.  That is,
\begin{equation}
\sum_{\tau \in S^k} \text{Wg}^U(\sigma^{-1}\tau, d) \, G^U(\tau^{-1} \pi, d) = \delta_{\sigma, \pi}\,.
\end{equation}

Having defined the Weingarten function, let us state a few useful results from the literature which we will leverage in our proofs:
\begin{CM3p2} For any $\sigma \in S_k$ and $d > \sqrt{6}\,k^{7/4}$,
\begin{equation}
\frac{1}{1 - \frac{k-1}{d}} \leq \frac{(-1)^{k - \#(\sigma)} d^{2k - \#(\sigma)} \text{\rm Wg}^U(\sigma, d)}{\prod_i \frac{(2\ell_i - 2)!}{(\ell_i - 1)! \ell_i!}} \leq \frac{1}{1 - \frac{6 k^{7/2}}{d^2}}\,,
\end{equation}
where the left-hand side inequality is valid for any $d \geq k$.  Here $\sigma \in S_k$ has cycle type $(\ell_1, \ell_2,...)$.
\end{CM3p2}
\noindent We will in fact use the following corollary of this result:
\begin{corollary}
\label{E:corr1}
$|\text{\rm Wg}^U(\mathds{1},d) - d^{-k}| \leq O(k^{7/2} d^{-(k+2)})$.
\end{corollary}
\noindent Finally we state a Lemma from~\cite{aharonov2021quantum}:
\begin{ACQl6}
$\sum_{\tau \in S_k} |\text{\rm Wg}^U(\tau, d)| = \frac{(d-k)!}{d!}$.
\end{ACQl6}

\subsection{Learning tree formalism}

In the main text, we provided definitions of time-ordered and out-of-time-order experiments based on the learning tree framework in quantum learning theory~\cite{aharonov2021quantum, chen2021exponential}.  It is useful to formalize these more precisely; our definitions below are closely based off of Definition 6.1 of~\cite{chen2021exponential}.
\begin{definition}[Tree representation for learning a collection of channels without a quantum memory] Let $S = \{\mathcal{C}_i\}_i$ be a set of quantum channels on states on $\mathcal{H}$.  A quantum learning algorithm without memory can be cast as a rooted tree $\mathcal{T}$ of depth $T$ where each vertex encodes all of the classical measurement outcomes that have been obtained by the algorithm up until then.  The tree $\mathcal{T}$ satisfies the following properties:
\begin{enumerate}
    \item Each note $u$ has an associated $n$-qubit unnormalized state $\rho^S(u)$ corresponding to the current state of the system.
    \item At the root $r$ of the tree, $\rho^S(r)$ is the initial state $\rho_0$.
    \item At each node $u$ (except the root node) we apply a POVM measurement $\{F_s^u\}_s$ on $\rho(u)$ to obtain a classical outcome $s$.  Without loss of generality we take all of the $F_s^u$'s to be rank one; if they are not, we can simply refine $\{F_s^u\}_s$ so that each of its elements is rank one.  We also have a function $f^u$ which takes the index set of $\{F_s^u\}_s$ to the index set of $S = \{\mathcal{C}_i\}_i$. Then we apply the channel $\mathcal{C}_{f^u(s)}$ to the present state. Each child node $v$ of $u$ is connected through the edge $e_{u,s}$.
    \item 
    If $v$ is the child node of $u$ connected through the edge $e_{u,s}$, then
    \begin{equation}
    \rho^S(v) := \mathcal{C}_{f^u(s)}[F_s^u\,\rho^S(u)\,(F_s^u)^\dagger]\,.
    \end{equation}
    Here $F_s^u\,\rho^S(u)\,(F_s^u)^\dagger$ is the unnormalized post-measurement state, to which the channel $\mathcal{C}_{f^u(s)}$ is applied.
    \item For any node $u$ at depth $t$ in the tree, $p^S(u) := \text{\rm tr}(\rho^S(u))$ is the probability that the transcript of measurement outcomes observed by the learning algorithm after $t$ measurements is $u$.  Moreover, $\rho^S(u)/p^S(u)$ is the state of the system at the node $u$.
\end{enumerate}
\end{definition}
\noindent Using this definition, we can provide the following formalizations of Definitions 1 and 2 in the main text:
\begin{definition}[Time-ordered experiment, formal]
A time-ordered experiment is a tree representation for learning a single unitary channel $\mathcal{U}$ without quantum memory.
\end{definition}
\begin{definition}[Out-of-time-order experiment, formal]
An out-of-time-order experiment is a tree representation for learning the collection of two unitary channels $\{\mathcal{U}, \mathcal{U}^\dagger\}$ without quantum memory, where the channels are inverses of one another.
\end{definition}

Suppose we have a tree representation $\mathcal{T}$ for learning a collection of channels without quantum memory, with depth $T$.  Let its associated collection of channels be $S = \{\mathcal{C}_i\}_i$.  Then the probability distribution over measurement outcomes is given by $p^S(\ell)$ where $\ell$ runs over the leafs of the tree.  If instead we had a collection of channels $S' = \{\mathcal{C}_i'\}_i$ with the same index set as $S = \{\mathcal{C}_i\}_i$, then we could run $S'$ through the same learning tree protocol so that the probability distribution over measurement outcomes is now $p^{S'}(\ell)$.  If we did not know if we were handed $S$ or $S'$, then Le Cam's two point method~\cite{yu1997assouad} implies that any post-processing algorithm we might use on our measurement data to distinguish between $S$ and $S'$ can succeed with a probability $p \geq 1/2$ only if
\begin{equation}
\frac{1}{2} \sum_{\ell\, \in \, \text{leaf}(\mathcal{T})} |p^{S}(\ell) - p^{S'}(\ell)| \geq 2p - 1\,.
\end{equation}

Circling back to Theorem 2, it is thus sufficient to show that for any time-ordered experiment corresponding to a learning tree $\mathcal{T}$ of depth $T$,
\begin{equation}
\label{E:toprove1}
\frac{1}{2} \sum_{\ell\, \in \, \text{leaf}(\mathcal{T})} |\mathbb{E}_{\mathcal{U}}[p^{\mathcal{U}}(\ell)] - \mathbb{E}_{\mathcal{U}_1, \mathcal{U}_2}[p^{\mathcal{U}_1 \otimes \mathcal{U}_2}(\ell)]| \leq \frac{2}{3}
\end{equation}
for $T \leq \Omega(d^{1/4})$. This inequality would show that we cannot distinguish between the two ensembles with success probability $p \geq 5/6$ using fewer than $\Omega(d^{1/4})$ queries to the unknown unitary. Here $\mathbb{E}_{\mathcal{U}}$ denotes the Haar average over $U$ in the unitary channel $\mathcal{U}[\rho] = U \rho U^\dagger$, and similarly for $\mathbb{E}_{\mathcal{U}_1,\mathcal{U}_2}$ and $(\mathcal{U}_1 \otimes \mathcal{U}_2)[\rho] = (U_1 \otimes U_2)\rho(U_1^\dagger \otimes U_2^\dagger)$.  We will prove the inequality~\eqref{E:toprove1} below.


\section{Main proofs}

As explained above, we can reformulate Theorem 2 in the following manner:
\begin{theorem}[Equivalent to Theorem 2]
For any time-ordered experiment with learning tree $\mathcal{T}$ with depth $T \leq \Omega(d^{1/4})$, we have
\begin{equation}
    \frac{1}{2} \sum_{\ell\, \in \, \text{\rm leaf}(\mathcal{T})} |\mathbb{E}_{\mathcal{U}}[p^{\mathcal{U}}(\ell)] - \mathbb{E}_{\mathcal{U}_1, \mathcal{U}_2}[p^{\mathcal{U}_1 \otimes \mathcal{U}_2}(\ell)]| \leq \frac{2}{3}\,.
\end{equation}
\end{theorem}
\begin{proof}
Let $\mathcal{D}$ be the maximally depolarizing channel so that $p^{\mathcal{D}}(\ell) = 1/d^T$.  Then using the triangle inequality,
\begin{equation}
\frac{1}{2} \sum_{\ell\, \in \, \text{\rm leaf}(\mathcal{T})} |\mathbb{E}_{\mathcal{U}}[p^{\mathcal{U}}(\ell)] - \mathbb{E}_{\mathcal{U}_1, \mathcal{U}_2}[p^{\mathcal{U}_1 \otimes \mathcal{U}_2}(\ell)]| \leq \frac{1}{2} \sum_{\ell\, \in \, \text{\rm leaf}(\mathcal{T})} |p^{\mathcal{D}}(\ell) - \mathbb{E}_{\mathcal{U}}[p^{\mathcal{U}}(\ell)]| + \frac{1}{2} \sum_{\ell\, \in \, \text{\rm leaf}(\mathcal{T})} |p^{\mathcal{D}}(\ell) - \mathbb{E}_{\mathcal{U}_1, \mathcal{U}_2}[p^{\mathcal{U}_1 \otimes \mathcal{U}_2}(\ell)]|\,.
\end{equation}
By Proposition 1 below, the first term on the right-hand side is less than or equal to $1/3$ for $T \leq \Omega(d^{1/3})$.  Similarly,  by Proposition 2 below the second term is less than or equal to $1/3$ for $T \leq \Omega(d^{1/4})$.  This completes the proof.
\end{proof}

\subsection{Unitary channel versus maximally depolarizing channel}

We begin by establishing more notation.  Given a learning tree $\mathcal{T}$ of depth $T$, let $v_0, v_1, ..., v_T = \ell$ be a root-to-leaf path through the tree.  This corresponds to having measured a sequence of POVM elements; let us denote them by $F_{v_1}, F_{v_2}, ..., F_{v_T}$.  Without loss of generality these can be assumed to be rank one, as we explained previously.  We can treat the last round (i.e.~the $T$th round) differently than all of the others, since we do not need to have a residual state after measuring.  This allows us to replace $F_{v_T}$ by a bra $\langle \psi_{v_T}|$.  Our $F_{v_i}$'s and $\langle \psi_{v_T}|$'s satisfy completeness relations, namely
\begin{equation}
\sum_{v \,\in\, \text{child}(v_{i-1})} F_{v}^\dagger F_{v} = \mathds{1} \qquad \text{for }i=1,...,T-1\,,
\end{equation}
and also
\begin{equation}
\sum_{v \,\in\, \text{child}(v_{T-1})} |\psi_{v}\rangle \langle \psi_v| = \mathds{1}\,.
\end{equation}
With this notation at hand, we can write $p^\mathcal{U}(\ell)$ as
\begin{equation}
p^{\mathcal{U}}(\ell) = \langle \psi_{v_T}| U  F_{v_{T-1}} U \cdots U F_{v_2}\rho_0 F_{v_1}^\dagger U^\dagger  \cdots U^\dagger F_{v_{T-1}} U^\dagger |\psi_{v_T}\rangle\,.
\end{equation}
For later, it will also be convenient to define $F_\ell$ by $F_\ell := F_{v_1} \otimes F_{v_2} \otimes \cdots \otimes F_{v_{T-1}}$. 

With these preparations in order, we turn to our desired Proposition:

\begin{proposition}
\label{Prop:prop1}
For $T \leq \Omega(d^{1/3})$, we have
\begin{equation}
\label{E:1normineq1}
\frac{1}{2}\sum_{\ell \,\in\, \text{\rm leaf}(\mathcal{T})} |p^{\mathcal{D}}(\ell) - \mathbb{E}_{\mathcal{U}}[ p^{\mathcal{U}}(\ell)]| \leq \frac{1}{3}\,.
\end{equation}
\end{proposition}
\begin{proof}
Using our Haar integration results from earlier, we have
\begin{align}
    \includegraphics[scale=.28, valign = c]{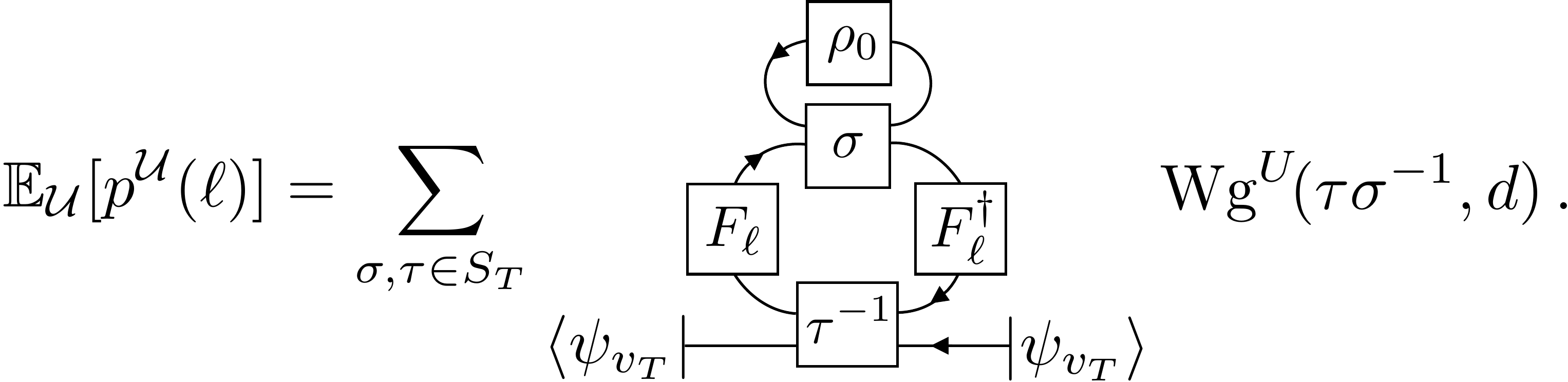} \nonumber
\end{align}
We let $p_{\sigma,\tau}(\ell)$ denote the summand of the above.  Now we can upper bound the left-hand side of~\eqref{E:1normineq1} using the triangle inequality and the Cauchy-Schwarz inequality as
\begin{align}
 \label{E:1normineq2}
\frac{1}{2}\sum_{\ell \,\in\, \text{\rm leaf}(\mathcal{T})} |p^{\mathcal{D}}(\ell) - \mathbb{E}_{\mathcal{U}}[ p^{\mathcal{U}}(\ell)]| &\leq \frac{1}{2}\sum_{\ell \,\in\, \text{\rm leaf}(\mathcal{T})} |p^{\mathcal{D}}(\ell) - p_{\mathds{1},\mathds{1}}(\ell)| + \frac{1}{2}\sum_{\ell \,\in\, \text{\rm leaf}(\mathcal{T})} \sum_{\sigma \not = \mathds{1}}| p_{\sigma,\mathds{1}}(\ell)|  \nonumber \\
& \qquad \qquad \qquad \qquad \qquad \qquad \qquad \qquad \qquad + \frac{1}{2}\sum_{\ell \,\in\, \text{\rm leaf}(\mathcal{T})} \sum_{\tau \not = \mathds{1},\,\sigma}|p_{\sigma, \tau}(\ell)|\,.
\end{align}
We will proceed by bounding each of the three terms on the right-hand side of~\eqref{E:1normineq2} in turn. \\ \\
\textbf{First term} \\ \\
For the first term, we can apply Cauchy-Schwarz to find the upper bound
\begin{align}
    \includegraphics[scale=.28, valign = c]{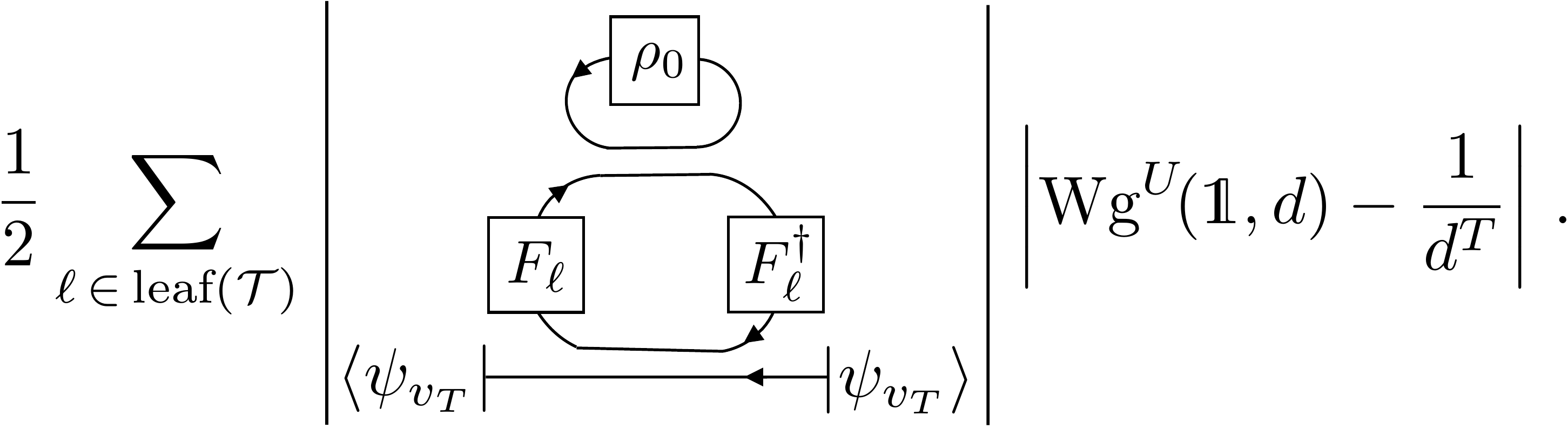} \nonumber
\end{align}
The absolute value in the first term can be removed since its argument is positive.  Now we can explicitly sum over leafs to obtain
\begin{equation}
\frac{d^T}{2}\left|\text{Wg}^{U}(\mathds{1},d) - \frac{1}{d^T}\right|\,.
\end{equation}
But using Corollary~\ref{E:corr1} we have $\left|\text{Wg}^{U}(\mathds{1},d) - \frac{1}{d^T}\right| \leq O(T^{7/2}/d^{T+2})$ for $T < \left(\frac{d}{\sqrt{6}}\right)^{4/7}$, and so in total
\begin{equation}
\frac{1}{2}\sum_{\ell \,\in\, \text{\rm leaf}(\mathcal{T})} |p^{\mathcal{D}}(\ell) - p_{\mathds{1},\mathds{1}}(\ell)| \leq O\!\left(\frac{T^{7/2}}{d^{2}}\right)\,.
\end{equation}
$$$$
\textbf{Second term}
\\ \\
Applying Cauchy-Schwarz to the second term on the right-hand side of~\eqref{E:1normineq2}, we have the upper bound
\begin{align}
\label{E:secondtermineq1}
    \includegraphics[scale=.28, valign = c]{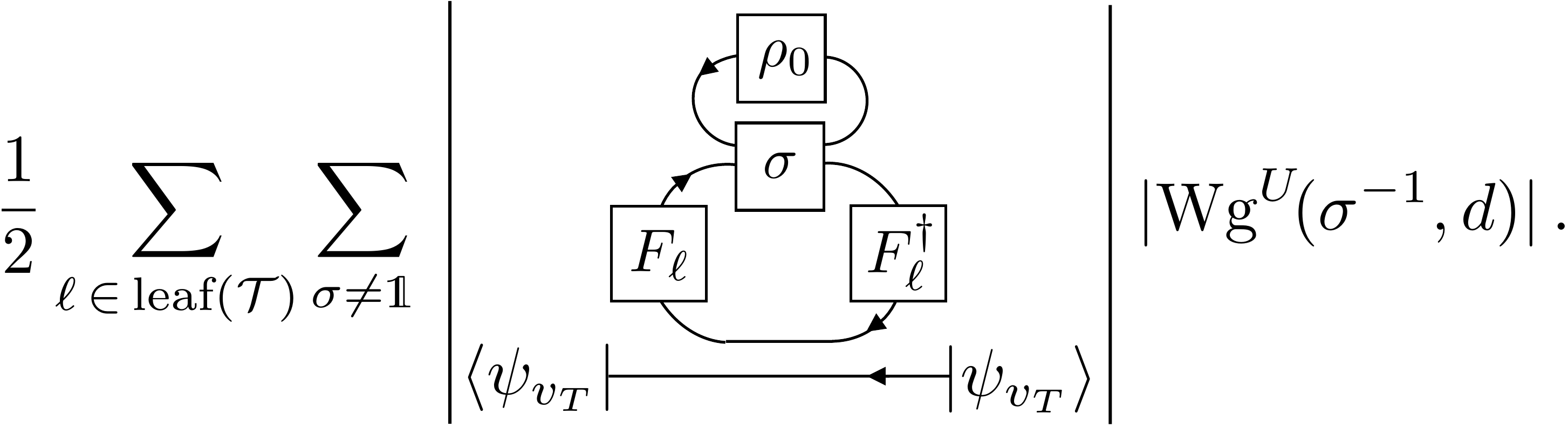} 
\end{align}
The first term can be upper bounded using H\"{o}lder's inequality
\begin{align}
    \includegraphics[scale=.28, valign = c]{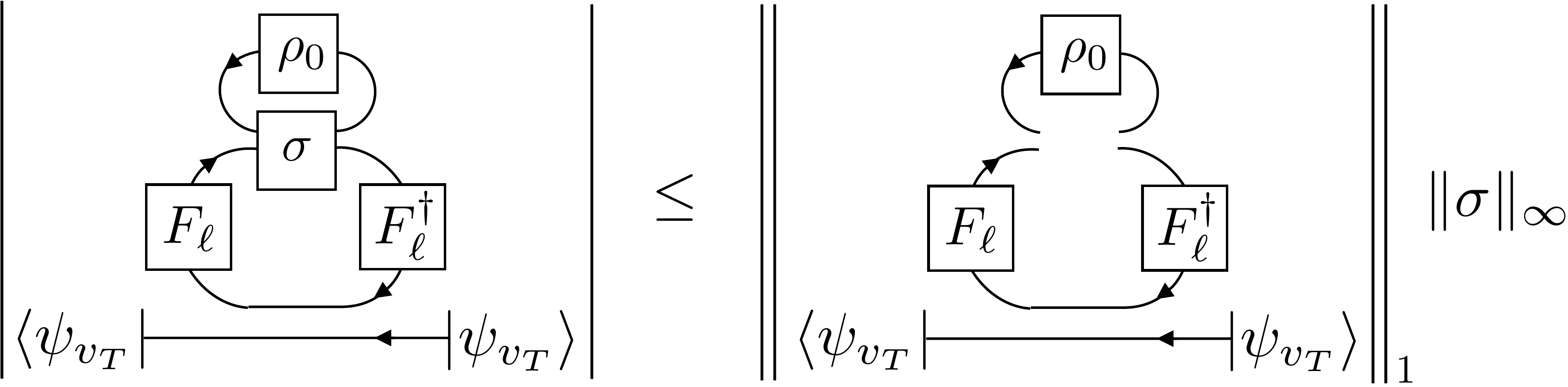} \nonumber
\end{align}
and we further use the equality
\begin{align}
\label{E:usefulmarker1}
    \includegraphics[scale=.28, valign = c]{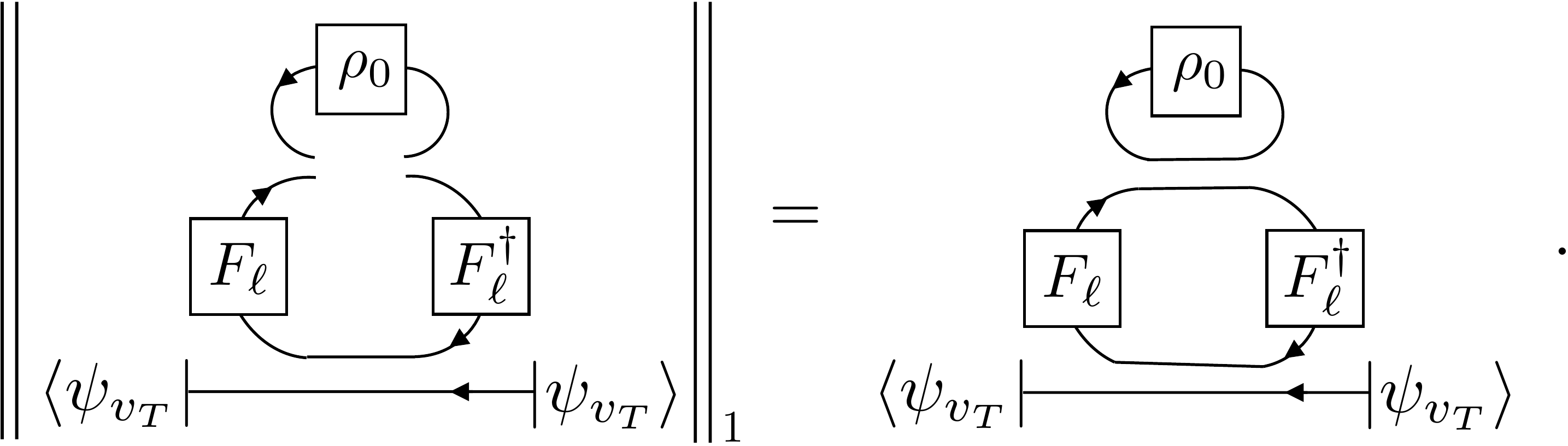}
\end{align}
This follows from the fact that $\|A \otimes B\|_1 = \|A\|_1 \|B\|_1 = \text{tr}(A) \,\text{tr}(B)$ if $A$ and $B$ are positive semi-definite.  Then~\eqref{E:secondtermineq1} is upper bounded by
\begin{align}
\label{E:usefulmarker2}
    \includegraphics[scale=.28, valign = c]{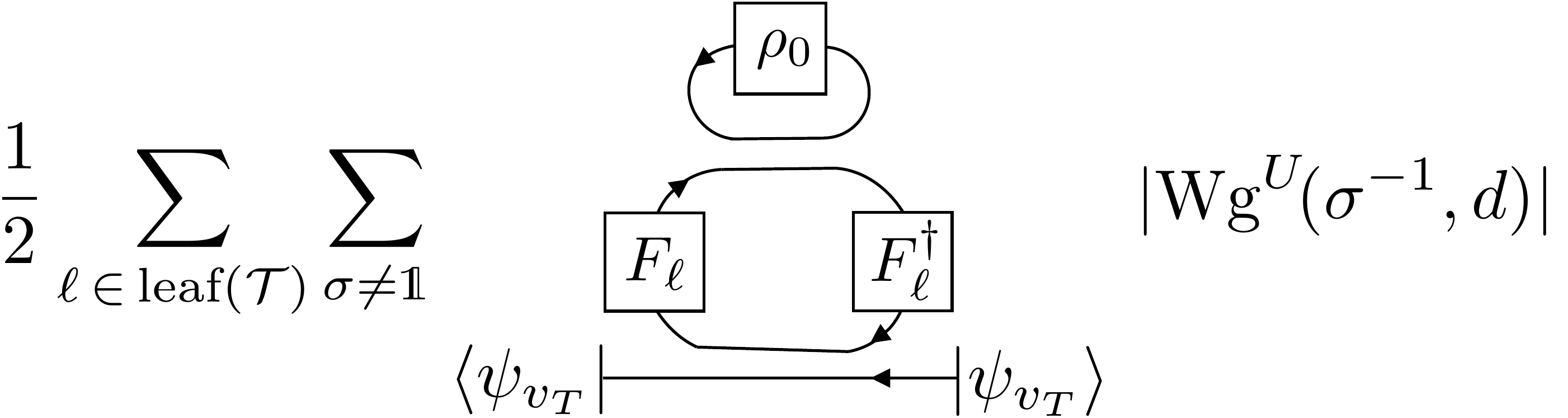}
\end{align}
and so summing over leafs we obtain
\begin{equation}
\frac{d^T}{2}\sum_{\sigma \not = \mathds{1}} |\text{Wg}^U(\sigma^{-1},d)|\,.
\end{equation}
But this quantity is less than or equal to $O(T^2/d)$ using Lemma 6 of~\cite{aharonov2021quantum}; thus we summarily find
\begin{equation}
 \frac{1}{2}\sum_{\ell \,\in\, \text{\rm leaf}(\mathcal{T})} \sum_{\sigma \not = \mathds{1}}| p_{\sigma,\mathds{1}}(\ell)|  \leq O\!\left(\frac{T^2}{d}\right)\,.
\end{equation}
$$$$
\textbf{Third term}
\\ \\
As usual, we apply the Cauchy-Schwarz inequality to the last term on the right-hand side of~\eqref{E:1normineq2} to obtain
\begin{align}
\label{E:thirdtermfirstfig}
    \includegraphics[scale=.28, valign = c]{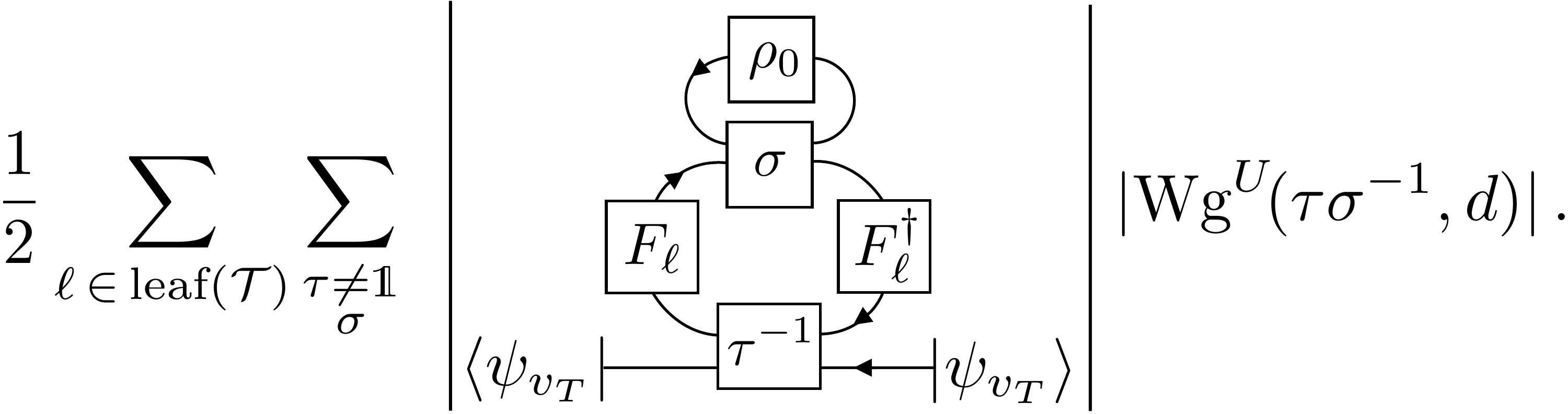}
\end{align}
Similar to the previous case, we apply H\"{o}lder's inequality to the diagrammatic term in the summand as
\begin{align}
    \includegraphics[scale=.28, valign = c]{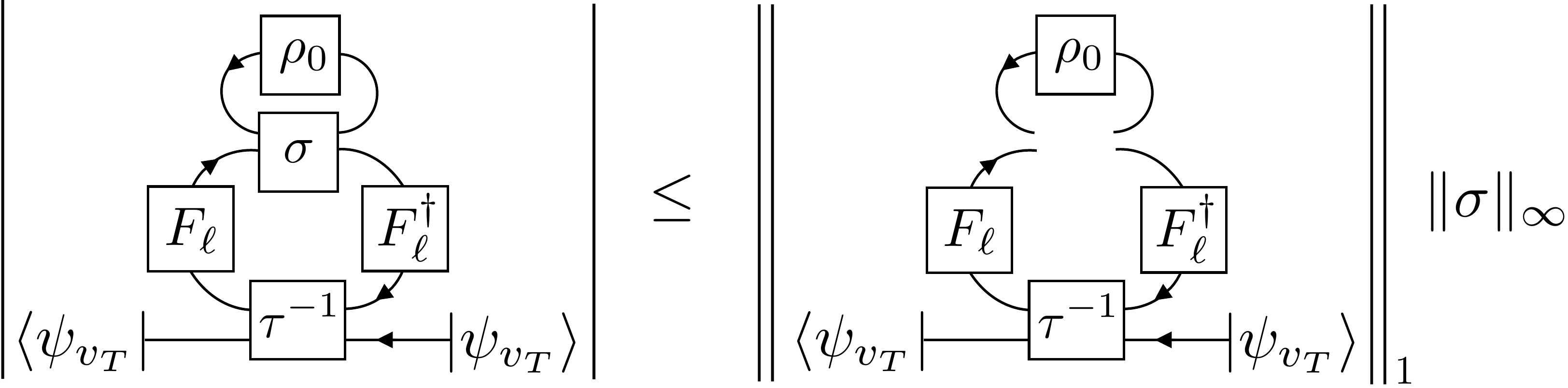}
\end{align}
noting again that $\|\sigma\|_\infty = 1$.  The other 1-norm term further simplifies to
\begin{align}
    \includegraphics[scale=.28, valign = c]{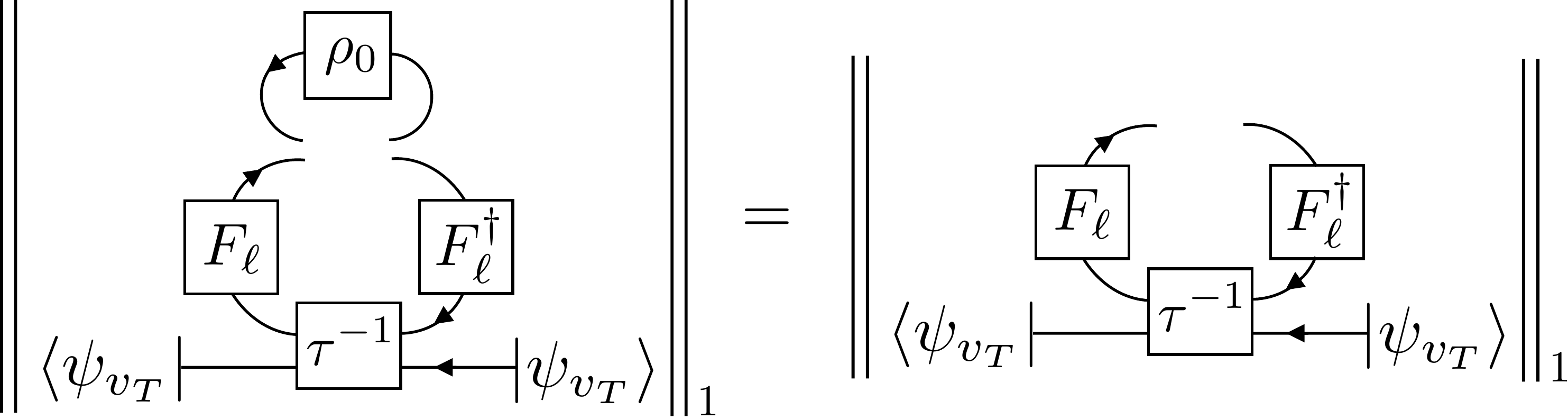}
\end{align}
where we have used $\|A \otimes B\|_1 = \|A\|_1 \|B\|_1$, where in the above setting $\|B\|_1 = 1$.
It is convenient to simplify the remaining 1-norm term for fixed $\tau^{-1}$.  To do so, we decompose $\tau^{-1}$ into cycles as $\tau^{-1} = C_1 C_2 \cdots C_{\#(\tau^{-1})}$, and will say that $i \to j$ belongs to the $m$th cycle $C_m$ if $C_m = (\cdots ij \cdots )$.  More generally we also say that $i \to j$ belongs to $\tau^{-1}$.  Further letting $v_0 = r, v_1,...,v_{T-1}, v_T = \ell$ be the root-to-leaf path terminating in $\ell$, we leverage the following lemma:
\begin{lemma}
If $\tau^{-1}$ contains the 1-cycle $T \to T$ then
\begin{align}
\label{E:1normproduct1}
    \includegraphics[scale=.28, valign = c]{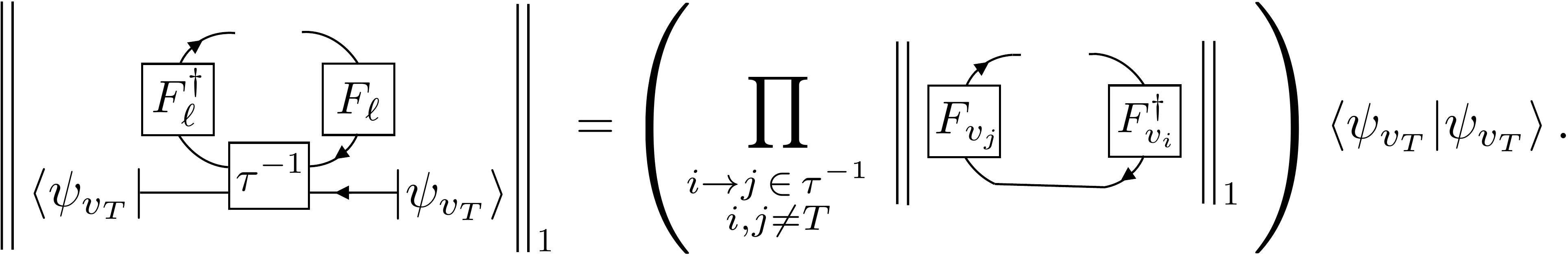}
\end{align}
Otherwise if $\tau^{-1}$ does not contain the 1-cycle $T \to T$, then it must contain some $\hat{i} \to T \to \hat{j}$ (where possibly $\hat{i} = \hat{j}$) in which case
\begin{align}
\label{E:1normproduct2}
    \includegraphics[scale=.28, valign = c]{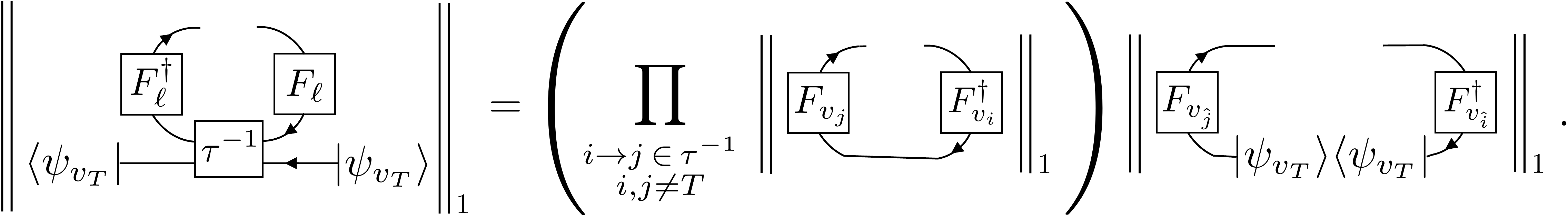}
\end{align}
\end{lemma}
\begin{proof}
The identities follow by contracting tensor indices of $F_\ell$, $F_\ell^\dagger$, $\langle \psi_{v_T}|$ and $|\psi_{v_T}\rangle$ according to $\tau^{-1}$, and then using the identity $\|A_1 \otimes A_2 \otimes \cdots \otimes A_n\|_1 = \prod_{i=1}^n \|A_i\|_1$.
\end{proof}
Next we simplify the 1-norm terms appearing in~\eqref{E:1normproduct1} and~\eqref{E:1normproduct2}. Since $\|A\|_1 = \|A\|_2$ when $A$ is rank one, we have
\begin{align}
    \includegraphics[scale=.28, valign = c]{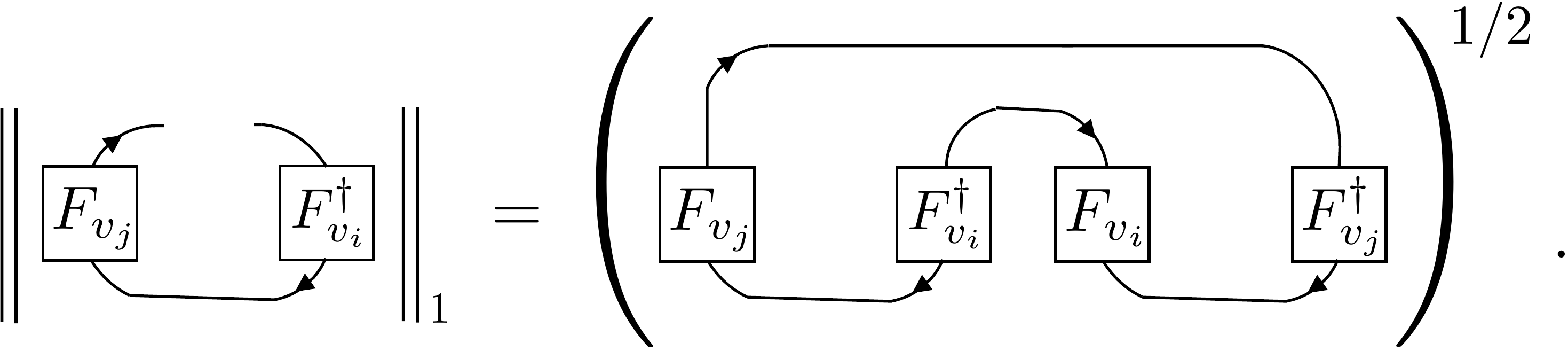}
\end{align}
Equivalently, this is
\begin{align}
    \includegraphics[scale=.28, valign = c]{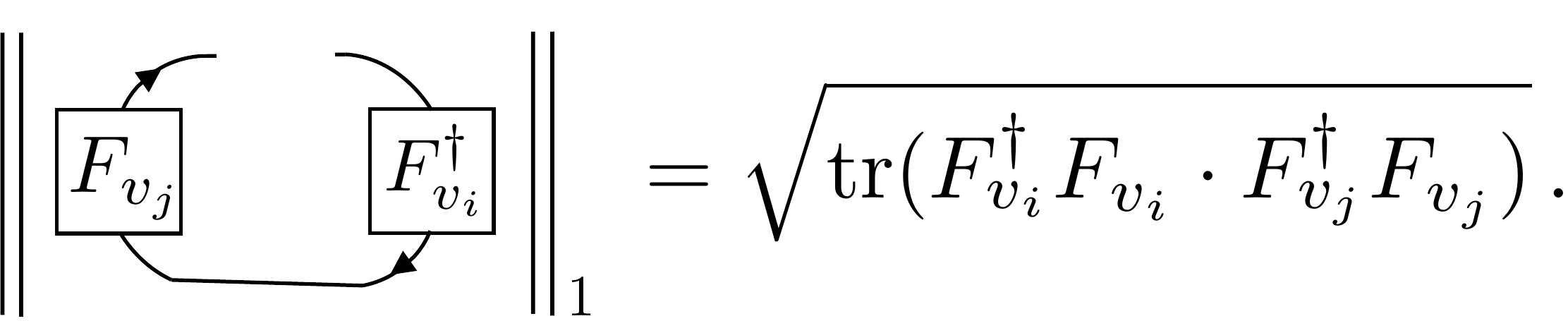}
\end{align}
Following the same logic, we obtain the equality
\begin{align}
    \includegraphics[scale=.28, valign = c]{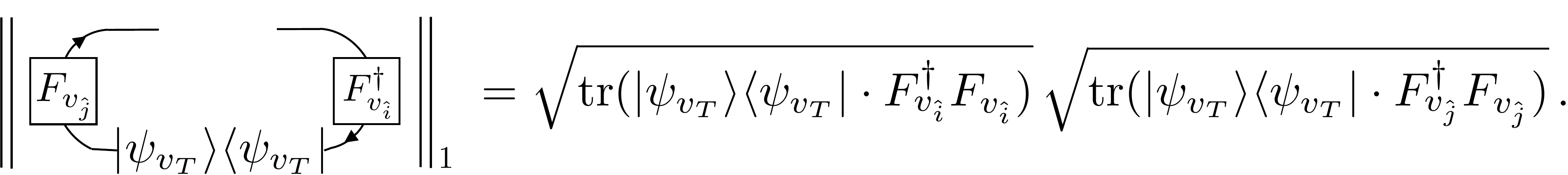}
\end{align}
Accordingly,~\eqref{E:1normproduct1} is equal to
\begin{equation}
\label{E:1normproduct3}
\left(\prod_{\substack{i \to j \in \tau^{-1} \\ i,j \not = T}} \sqrt{\text{tr}(F_{v_i}^\dagger F_{v_i}\cdot F_{v_j}^\dagger F_{v_j})}\right) \langle \psi_{v_T}|\psi_{v_T}\rangle
\end{equation}
and likewise~\eqref{E:1normproduct2} is equal to
\begin{equation}
\label{E:1normproduct4}
\left(\prod_{\substack{i \to j \in \tau^{-1} \\i,j \not = T}} \sqrt{\text{tr}(F_{v_i}^\dagger F_{v_i} \cdot F_{v_j}^\dagger F_{v_j})}\right) \sqrt{\text{tr}(|\psi_{v_T}\rangle \langle \psi_{v_T}| \cdot F_{v_{\hat{i}}}^\dagger F_{v_{\hat{i}}})} \,\sqrt{\text{tr}(|\psi_{v_T}\rangle \langle \psi_{v_T}| \cdot F_{v_{\hat{j}}}^\dagger F_{v_{\hat{j}}})}\,.
\end{equation}
To make the two cases look the same, we simply define $F_{v_T} := |\psi_{v_T}\rangle \langle \psi_{v_T}|$ so that both~\eqref{E:1normproduct3} and~\eqref{E:1normproduct4} can be written as
\begin{equation}
\label{E:niceproduct1}
\prod_{i \to j \in \tau^{-1}} \sqrt{\text{tr}(F_{v_i}^\dagger F_{v_i} \cdot F_{v_j}^\dagger F_{v_j})}\,.
\end{equation}
To further bound~\eqref{E:niceproduct1}, we decompose the product into cycles as
\begin{equation}
\label{E:cycledecomp1}
\prod_{m = 1}^{\#(\tau^{-1})}\prod_{i \to j \in C_m} \sqrt{\text{tr}(F_{v_i}^\dagger F_{v_i} \cdot F_{v_j}^\dagger F_{v_j})}
\end{equation}
and treat each $\prod_{i \to j \in C_m} \sqrt{\text{tr}(F_{v_i}^\dagger F_{v_i} \cdot F_{v_j}^\dagger F_{v_j})}$ term separately.

Suppose that $C_m$ is a cycle of length $p$ (often denoted by $|C_m| =  p$); then we can write it as $C_m = (a_{m,1} a_{m,2} \cdots a_{m,p})$ where $\{a_{m,1}, a_{m,2},...,a_{m,p}\} \subseteq \{1,2,...,T\}$.  Then we can write $\prod_{i \to j \in C_m} \sqrt{\text{tr}(F_{v_i}^\dagger F_{v_i} \cdot F_{v_j}^\dagger F_{v_j})}$ as
\begin{equation}
\label{E:cyclenewlabels1}
\prod_{i = 1}^p \sqrt{\text{tr}(F_{v_{a_{m,i}}}^\dagger F_{v_{a_{m,i}}} \cdot F_{v_{a_{m,i+1}}}^\dagger F_{v_{a_{m,i+1}}})}
\end{equation}
where the $i$ subscripts are treated modulo $p$.  Consider two cases: \\ \\
\textbf{Case 1:} $p$\textit{ is even}.  We can split up~\eqref{E:cyclenewlabels1} into two products as
\begin{equation}
\left(\prod_{i\text{ odd}} \sqrt{\text{tr}(F_{v_{a_{m,i}}}^\dagger F_{v_{a_{m,i}}} \cdot F_{v_{a_{m,i+1}}}^\dagger F_{v_{a_{m,i+1}}})} \right) \left(\prod_{j\text{ even}} \sqrt{\text{tr}(F_{v_{a_{m,j}}}^\dagger F_{v_{a_{m,j}}} \cdot F_{v_{a_{m,j+1}}}^\dagger F_{v_{a_{m,j+1}}})} \right)   
\end{equation}
and using the inequality $ab \leq \frac{1}{2}(a^2 + b^2)$ we obtain the upper bound
\begin{equation}
\frac{1}{2} \prod_{i\text{ odd}} \text{tr}(F_{v_{a_{m,i}}}^\dagger F_{v_{a_{m,i}}} \cdot F_{v_{a_{m,i+1}}}^\dagger F_{v_{a_{m,i+1}}}) + \frac{1}{2}\prod_{j\text{ even}} \text{tr}(F_{v_{a_{m,j}}}^\dagger F_{v_{a_{m,j}}} \cdot F_{v_{a_{m,j+1}}}^\dagger F_{v_{a_{m,j+1}}})\,.   
\end{equation}
Let us define the first term as $\frac{1}{2}\, R_{m,-}$ and the second term as $\frac{1}{2}\,R_{m,+}$. \\ \\
\textbf{Case 2:} $p$\textit{ is odd}.  Here we opt to split up~\eqref{E:cyclenewlabels1} as
\begin{align}
&\sqrt{\text{tr}(F_{v_{a_{m,p}}}^\dagger F_{v_{a_{m,p}}} \cdot F_{v_{a_{m,1}}}^\dagger F_{v_{a_{m,1}}})}\left(\prod_{\substack{i\text{ odd} \\ 1 \leq i \leq p - 2}} \sqrt{\text{tr}(F_{v_{a_{m,i}}}^\dagger F_{v_{a_{m,i}}} \cdot F_{v_{a_{m,i+1}}}^\dagger F_{v_{a_{m,i+1}}})} \right) \nonumber \\
& \qquad \qquad \qquad \qquad \qquad \qquad \qquad \qquad \qquad \qquad \qquad \qquad \times \left(\prod_{j\text{ even}} \sqrt{\text{tr}(F_{v_{a_{m,j}}}^\dagger F_{v_{a_{m,j}}} \cdot F_{v_{a_{m,j+1}}}^\dagger F_{v_{a_{m,j+1}}})} \right)       
\end{align}
For $A,B$ positive semi-definite we have $\text{tr}(A B) \leq \|A\|_2 \|B\|_2 \leq \|A\|_1 \|B\|_1 \leq \text{tr}(A) \text{tr}(B)$ and so $\sqrt{\text{tr}(F_{v_{a_{m,p}}}^\dagger F_{v_{a_{m,p}}} \cdot F_{v_{a_{m,1}}}^\dagger F_{v_{a_{m,1}}})} \leq \sqrt{\text{tr}(F_{v_{a_{m,p}}}^\dagger F_{v_{a_{m,p}}})}\, \sqrt{\text{tr}(F_{v_{a_{m,1}}}^\dagger F_{v_{a_{m,1}}})}$\,.  Then the above equation is upper bounded by
\begin{align}
&\sqrt{\text{tr}(F_{v_{a_{m,p}}}^\dagger F_{v_{a_{m,p}}})}\, \sqrt{\text{tr}(F_{v_{a_{m,1}}}^\dagger F_{v_{a_{m,1}}})}\left(\prod_{\substack{i\text{ odd} \\ 1 \leq i \leq p - 2}} \sqrt{\text{tr}(F_{v_{a_{m,i}}}^\dagger F_{v_{a_{m,i}}} \cdot F_{v_{a_{m,i+1}}}^\dagger F_{v_{a_{m,i+1}}})} \right) \nonumber \\
& \qquad \qquad \qquad \qquad \qquad \qquad \qquad \qquad \qquad \qquad \qquad \qquad \times \left(\prod_{j\text{ even}} \sqrt{\text{tr}(F_{v_{a_{m,j}}}^\dagger F_{v_{a_{m,j}}} \cdot F_{v_{a_{m,j+1}}}^\dagger F_{v_{a_{m,j+1}}})} \right)       
\end{align}
and so using $ab \leq \frac{1}{2}(a^2 + b^2)$ we have the further upper bound
\begin{align}
& \frac{1}{2}\,\text{tr}(F_{v_{a_{m,p}}}^\dagger F_{v_{a_{m,p}}}) \prod_{\substack{i\text{ odd} \\ 1 \leq i \leq p - 2}} \text{tr}(F_{v_{a_{m,i}}}^\dagger F_{v_{a_{m,i}}} \cdot F_{v_{a_{m,i+1}}}^\dagger F_{v_{a_{m,i+1}}}) \nonumber \\
&\qquad \qquad \qquad \qquad \qquad \qquad \qquad \qquad + \frac{1}{2}\,\text{tr}(F_{v_{a_{m,1}}}^\dagger F_{v_{a_{m,1}}}) \prod_{j\text{ even}} \text{tr}(F_{v_{a_{m,j}}}^\dagger F_{v_{a_{m,j}}} \cdot F_{v_{a_{m,j+1}}}^\dagger F_{v_{a_{m,j+1}}})\,.  
\end{align}
We similarly call the first term $\frac{1}{2}\, R_{m,-}$ and the second term $\frac{1}{2}\,R_{m,+}$.
\\ \\
Taken together, Case 1 and Case 2 give us the following bound on~\eqref{E:cycledecomp1}:
\begin{align}
\prod_{m = 1}^{\#(\tau^{-1})}\prod_{i \to j \in C_m} \sqrt{\text{tr}(F_{v_i}^\dagger F_{v_i} \cdot F_{v_j}^\dagger F_{v_j})} &\leq \frac{1}{2^{\#(\tau^{-1})}} \prod_{m=1}^{\#(\tau^{-1})} (R_{m,-} + R_{m,+}) \nonumber \\
&= \frac{1}{2^{\#(\tau^{-1})}} \sum_{i_1,...,i_{\#(\tau^{-1})}= \pm} R_{1,i_1} R_{2,i_2} \cdots R_{\#(\tau^{-1}), i_{\#(\tau^{-1})}}\,.
\end{align}
Since the $R_{m,\pm}$'s depend implicitly on the leaf $\ell$, we add an $\ell$ superscript as $R_{m,\pm}^\ell$ to make the dependence explicit.  The summand $R^\ell_{1,i_1} R^\ell_{2,i_2} \cdots R^\ell_{\#(\tau^{-1}), i_{\#(\tau^{-1})}}$ for fixed indices $i_1, i_2,..., i_{\#(\tau^{-1})}$ has the feature that each $F_{v_i}^\dagger F_{v_i}$ for $i=1,...,T$ appears exactly once.  By virtue of this fact we can establish the following lemma:
\begin{lemma}
For any fixed set of indices $i_1, i_2,..., i_{\#(\tau^{-1})} \in \{+,-\}$, we have
\begin{equation}
\label{E:Rsum1}
\sum_{\ell \,\in\, \text{\rm leaf}(\mathcal{T})}  \,R^\ell_{1,i_1} R^\ell_{2,i_2} \cdots R^\ell_{\#(\tau^{-1}), i_{\#(\tau^{-1})}} \leq d^{T - \left\lfloor \frac{L(\tau^{-1})}{2} \right\rfloor}
\end{equation}
where $L(\tau^{-1})$ is the length of the longest cycle in $\tau^{-1}$.
\end{lemma}
\begin{proof}
We have the identities
\begin{align}
\label{E:usefulidentityminus1}
\sum_{v \, \in \, \text{child}(v_{i-1})} \text{tr}(F_{v}^\dagger F_{v} \cdot F_{v_j}^\dagger F_{v_j}) &= \text{tr}(F_{v_j}^\dagger F_{v_j})  \\
\label{E:usefulidentity0}
\sum_{v \, \in \, \text{child}(v_{i-1})} \text{tr}(F_{v}^\dagger F_{v}) &= d\,.
\end{align}
In a slight abuse of notation, we rewrite these as
\begin{align}
\label{E:usefulidentity1}
\sum_{v_i} \text{tr}(F_{v_i}^\dagger F_{v_i} \cdot F_{v_j}^\dagger F_{v_j}) &= \text{tr}(F_{v_j}^\dagger F_{v_j})  \\
\label{E:usefulidentity2}
\sum_{v_i} \text{tr}(F_{v_i}^\dagger F_{v_i}) &= d\,.
\end{align}
For fixed $i_1, i_2,..., i_{\#(\tau^{-1})} \in \{+,-\}$, we have
\begin{equation}
R^\ell_{1,i_1} R^\ell_{2,i_2} \cdots R^\ell_{\#(\tau^{-1}), i_{\#(\tau^{-1})}} = \left(\prod_{i \in \mathcal{S}_1} \text{tr}(F_{v_i}^\dagger F_{v_i})\right) \left(\prod_{(j,j') \in \mathcal{S}_2} \text{tr}(F_{v_j}^\dagger F_{v_j} \cdot F_{v_{j'}}^\dagger F_{v_{j'}})\right)
\end{equation}
where $\mathcal{S}_1 \subset \{1,...,T\}$ is the set of indices for which a $\text{tr}(F_{v_i}^\dagger F_{v_i})$ term appears, and $\mathcal{S}_2 \subset \{1,...,T\} \times \{1,...,T\}$ is the set of unordered pairs $(j,j')$ for which a $\text{tr}(F_{v_j}^\dagger F_{v_j} \cdot F_{v_{j'}}^\dagger F_{v_{j'}})$ term appears.  Note that the size of $\mathcal{S}_1$ is the number of odd-length cycles of $\tau^{-1}$.  As noted above, each $v_i$ for $i=1,...,T$ appears exactly once in the above expression.  Writing the~\eqref{E:Rsum1} as
\begin{equation}
\sum_{v_1}  \cdots \sum_{v_{T-1}} \sum_{v_T} \left(\prod_{i \in \mathcal{S}_1} \text{tr}(F_{v_i}^\dagger F_{v_i})\right) \left(\prod_{(j,j') \in \mathcal{S}_2} \text{tr}(F_{v_j}^\dagger F_{v_j} \cdot F_{v_{j'}}^\dagger F_{v_{j'}})\right)\,,
\end{equation}
we can perform the inner-most sum over $v_T$ following by the $v_{T-1}$ sum, and so on through the $v_1$ sum.  That is, we are summing from the leafs of the tree back up to the root; this order of summation is necessitated because of the adaptive nature of the measurement strategies that we allow.  That is, the choice of measurements in the future (i.e.~higher depth in the learning tree) depend on measurements made in the past (i.e.~lower depth in the learning tree), but not conversely.  Leveraging the identities~\eqref{E:usefulidentity1},~\eqref{E:usefulidentity2} and the equality $\left\lceil \frac{x}{2} \right\rceil + \left\lfloor \frac{x}{2} \right\rfloor = x$ for integer $x$, we find that the sum equals
\begin{equation}
d^{\sum_{m=1}^{\#(\tau^{-1})} \left\lceil \frac{|C_m|}{2} \right\rceil} = d^{T - \sum_{m=1}^{\#(\tau^{-1})} \left\lfloor \frac{|C_m|}{2} \right\rfloor} \leq d^{T - \left\lfloor \frac{L(\tau^{-1})}{2}\right\rfloor}
\end{equation}
giving the desired bound.
\end{proof}
An immediate consequence of the above Lemma is that
\begin{align}
\label{E:ineqtofloor1}
&\sum_{\ell \, \in \, \text{leaf}(\mathcal{T})} \left(\frac{1}{2^{\#(\tau^{-1})}} \sum_{i_1,...,i_{\#(\tau^{-1})}= \pm} R^\ell_{1,i_1} R^\ell_{2,i_2} \cdots R^\ell_{\#(\tau^{-1}), i_{\#(\tau^{-1})}}\right) \nonumber \\
=\,\,& \frac{1}{2^{\#(\tau^{-1})}} \sum_{i_1,...,i_{\#(\tau^{-1})}= \pm} \left(\sum_{\ell \, \in \, \text{leaf}(\mathcal{T})} R^\ell_{1,i_1} R^\ell_{2,i_2} \cdots R^\ell_{\#(\tau^{-1}), i_{\#(\tau^{-1})}}\right) \nonumber \\
\leq\,\,& \frac{1}{2^{\#(\tau^{-1})}} \sum_{i_1,...,i_{\#(\tau^{-1})}= \pm} d^{T - \left\lfloor \frac{L(\tau^{-1})}{2}\right\rfloor} \nonumber \\
=\,\,& d^{T - \left\lfloor \frac{L(\tau^{-1})}{2}\right\rfloor}\,.
\end{align}
Circling back to~\eqref{E:thirdtermfirstfig}, we can combine our bounds to obtain
\begin{equation}
\label{E:usefullater1}
\frac{1}{2}\sum_{\ell \,\in\,\text{leaf}(\mathcal{T})} \sum_{\substack{\tau \not = \mathds{1} \\ \sigma}} |p_{\sigma, \tau}(\ell)| \leq \frac{d^{T}}{2} \sum_{\sigma} |\text{Wg}^U(\sigma^{-1}, d)| \sum_{\tau \not = \mathds{1}} d^{- \left\lfloor \frac{L(\tau^{-1})}{2} \right\rfloor}\,.
\end{equation}
To bound the right-hand side, we can use $d^{T} \sum_{\sigma} |\text{Wg}^U(\sigma^{-1}, d)|  \leq 1 + O(T^2/d)$.  Letting $N(T,L)$ denote the number of permutations of $S_T$ whose longest cycle has length $L$, we can write
\begin{equation}
\sum_{\tau \not = \mathds{1}} d^{- \left\lfloor \frac{L(\tau^{-1})}{2} \right\rfloor} = \sum_{L=2}^T N(T,L) \, d^{- \left\lfloor \frac{L}{2}\right\rfloor}\,.
\end{equation}
Here the $L = 1$ case is omitted since this corresponds to the identity permutation.  Since $N(T,L) \leq \binom{T}{L} L! = \frac{T!}{(T-L)!} < T^L$, the above sum is upper bounded by
\begin{equation}
\sum_{L = 2}^\infty T^L d^{- \left\lfloor \frac{L}{2}\right\rfloor} = \frac{(1 + T) \frac{T^2}{d}}{1 - \frac{T^2}{d}} = \frac{T^3}{d} + \frac{T^2}{d} + O\left(\frac{T^5}{d^2}\right)\,.
\end{equation}
In summary, if $T \leq o(d^{1/3})$ we have
\begin{equation}
\frac{1}{2} \sum_{\ell \, \in \, \text{leaf}(\mathcal{T})} \sum_{\substack{\tau \not = \mathds{1} \\ \sigma}} |p_{\sigma, \tau}(\ell)| \leq o(1)\,,
\end{equation}
as needed.
\\ \\
Combining the three cases, we find that for $T \leq \Omega(d^{1/3})$ we have
\begin{equation}
\frac{1}{2}\sum_{\ell \,\in\, \text{\rm leaf}(\mathcal{T})} |p^{\mathcal{D}}(\ell) - (\mathbb{E}_{\mathcal{U}} \,p^{\mathcal{U}}(\ell))| \leq \frac{1}{3}
\end{equation}
which completes the proof.
\end{proof}

\subsection{Product unitary channel versus maximally depolarizing channel}

For the next proposition we use the same notations and conventions as we did above.

\begin{proposition}
For $T \leq \Omega(d^{1/4})$, we have
\begin{equation}
\label{E:1normineq1p2}
\frac{1}{2}\sum_{\ell \,\in\, \text{\rm leaf}(\mathcal{T})} |p^{\mathcal{D}}(\ell) - \mathbb{E}_{\mathcal{U}_1, \mathcal{U}_2}[p^{\mathcal{U}_1 \otimes \mathcal{U}_2}(\ell)]| \leq \frac{1}{3}\,.
\end{equation}
\end{proposition}
\begin{proof}
As before $p^{\mathcal{D}}(\ell) = 1/d^T$, and now we have
\begin{align}
    \includegraphics[scale=.28, valign = c]{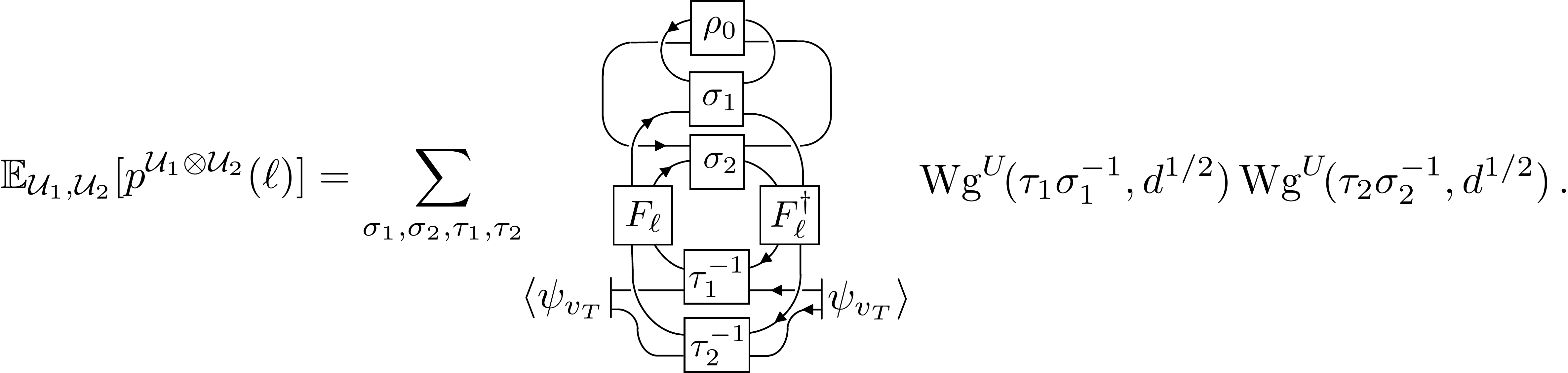}
\end{align}
We let $p_{\sigma_1, \sigma_2,\tau_1, \tau_2}(\ell)$ denote the summand of the above.  The left-hand side of~\eqref{E:1normineq1p2} can be upper bounded via the triangle and Cauchy-Schwarz inequalities as
\begin{align}
 \label{E:1normineq2p2}
\frac{1}{2}\sum_{\ell \,\in\, \text{\rm leaf}(\mathcal{T})} |p^{\mathcal{D}}(\ell) - \mathbb{E}_{\mathcal{U}_1,\mathcal{U}_2}[ p^{\mathcal{U}_1 \otimes \mathcal{U}_2}(\ell)]| &\leq \frac{1}{2}\sum_{\ell \,\in\, \text{\rm leaf}(\mathcal{T})} |p^{\mathcal{D}}(\ell) - p_{\mathds{1},\mathds{1},\mathds{1},\mathds{1}}(\ell)| + \frac{1}{2}\sum_{\ell \,\in\, \text{\rm leaf}(\mathcal{T})} \sum_{\sigma_1 \otimes \sigma_2 \not = \mathds{1} \otimes \mathds{1}}| p_{\sigma_1, \sigma_2, \mathds{1},\mathds{1}}(\ell)|  \nonumber \\
&\qquad \qquad \qquad \qquad \qquad \quad \quad + \frac{1}{2}\sum_{\ell \,\in\, \text{\rm leaf}(\mathcal{T})} \sum_{\substack{\tau_1 \otimes \tau_2 \not = \mathds{1} \otimes \mathds{1} \\ \sigma_1, \sigma_2}}|p_{\sigma_1, \sigma_2, \tau_1, \tau_2}(\ell)|\,.
\end{align}
Similar to the previous Proposition, we will individually bound each term on the right-hand side of~\eqref{E:1normineq2p2}. \\ \\
\textbf{First term} \\ \\
Applying the Cauchy-Schwarz inequality we have the upper bound
\begin{align}
    \includegraphics[scale=.28, valign = c]{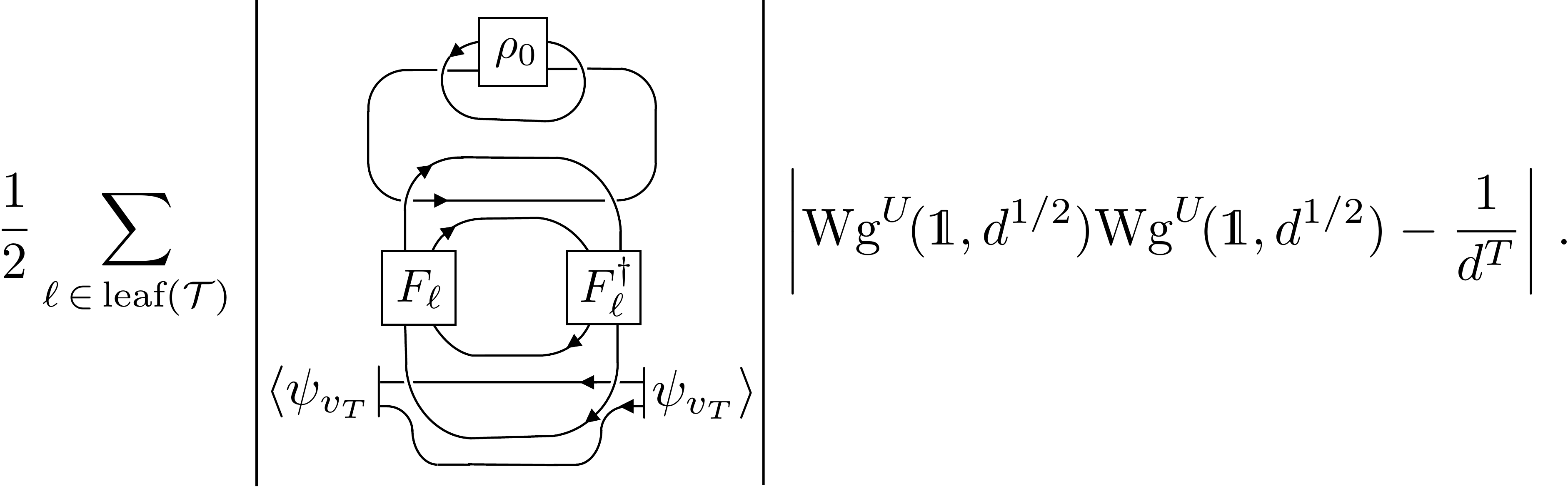}
\end{align}
Since the argument of the first term is positive, we remove the absolute values; this allows us to explicitly sum over leafs to obtain
\begin{equation}
\frac{d^T}{2}\left|\text{Wg}^{U}(\mathds{1},d^{1/2}) \text{Wg}^{U}(\mathds{1},d^{1/2}) - \frac{1}{d^T}\right|\,.
\end{equation}
The appropriate version of Corollary~\ref{E:corr1} gives us $\left|\text{Wg}^{U}(\mathds{1},d^{1/2}) - \frac{1}{d^{T/2}}\right| \leq O(T^{7/2}/d^{T/2+1})$ for $T < \left(\frac{d}{6}\right)^{2/7}$, and thus
\begin{equation}
\frac{1}{2}\sum_{\ell \,\in\, \text{\rm leaf}(\mathcal{T})} |p^{\mathcal{D}}(\ell) - p_{\mathds{1},\mathds{1},\mathds{1},\mathds{1}}(\ell)| \leq O\!\left(\frac{T^{7/2}}{d}\right)\,.
\end{equation}
$$$$
\textbf{Second term}
\\ \\
We can apply the Cauchy-Schwarz inequality to the second term on the right-hand side of~\eqref{E:1normineq2p2} to get the upper bound
\begin{align}
\label{E:secondtermineq1p2}
    \includegraphics[scale=.28, valign = c]{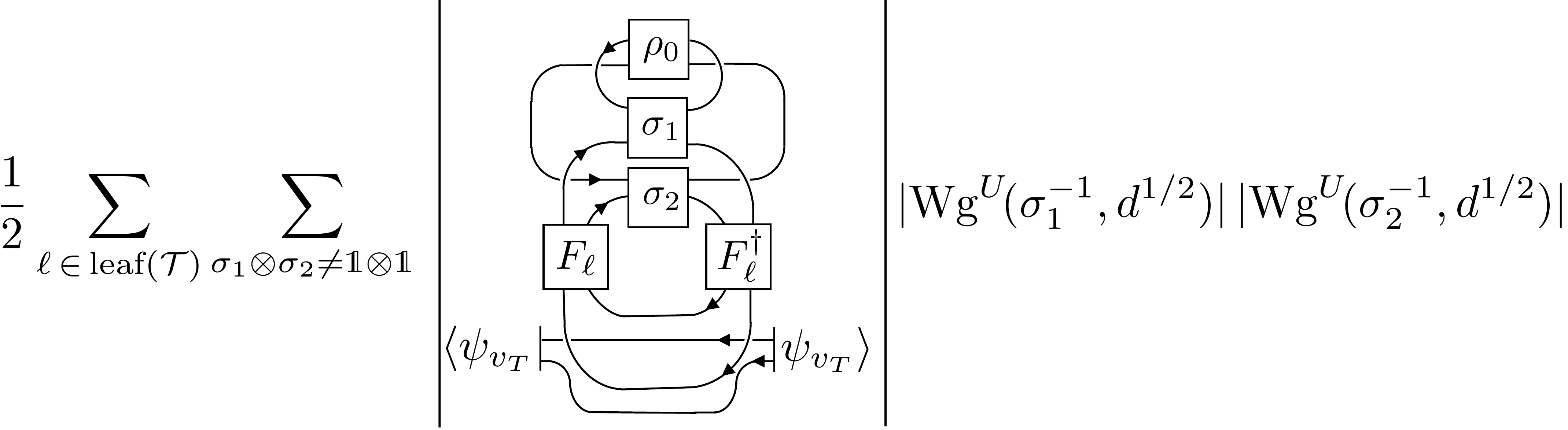}
\end{align}
The above can be bounded in the same manner as~\eqref{E:secondtermineq1} in Proposition~\ref{Prop:prop1}; the proof is the same up through~\eqref{E:usefulmarker1}.  Then the analog of~\eqref{E:usefulmarker2} is
\begin{align}
    \includegraphics[scale=.28, valign = c]{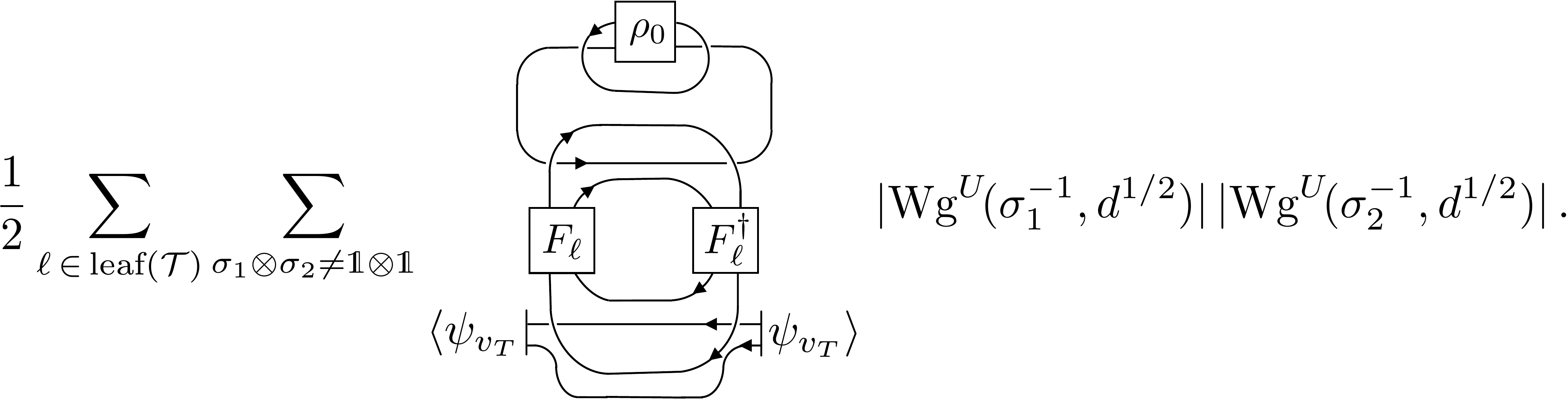}
\end{align}
and summing over leafs we find
\begin{equation}
\frac{d^T}{2}\sum_{\sigma_1 \otimes \sigma_2 \not = \mathds{1} \otimes \mathds{1}} |\text{Wg}^U(\sigma_1^{-1},d^{1/2})|\,|\text{Wg}^U(\sigma_2^{-1},d^{1/2})|\,.
\end{equation}
Denoting $\sigma = \sigma_1 \otimes \sigma_2$, the above is equal to
\begin{align}
2 \cdot \frac{d^{T/2}}{2} \,|\text{Wg}^U(\mathds{1},d^{1/2})|\left(d^{T/2} \sum_{\substack{\sigma \not = \mathds{1} \\ \sigma \in S_T}}|\text{Wg}^U(\sigma^{-1},d^{1/2})|\right) + \frac{1}{2}\left(d^{T/2}\sum_{\substack{\sigma \not = \mathds{1} \\ \sigma \in S_T}}|\text{Wg}^U(\sigma^{-1},d^{1/2})|\right)^2\,.
\end{align}
Since $|\text{Wg}^U(\mathds{1},d^{1/2})| = \frac{1}{d^{T/2}} + O(T^{7/2}/d^{T/2 + 1})$ for $T < \left(\frac{d}{6}\right)^{2/7}$ and the term in the parentheses is less than or equal to $O(T^2/d^{1/2})$ by Lemma 6 of~\cite{aharonov2021quantum}, we have in total
\begin{equation}
 \frac{1}{2}\sum_{\ell \,\in\, \text{\rm leaf}(\mathcal{T})} \sum_{\sigma_1 \otimes \sigma_2 \not = \mathds{1} \otimes \mathds{1}}| p_{\sigma_1, \sigma_2,\mathds{1}, \mathds{1}}(\ell)|  \leq O\!\left(\frac{T^2}{d^{1/2}}\right)\,.
\end{equation}
$$$$
\textbf{Third term}
\\ \\
Applying the Cauchy-Schwarz inequality to final term in~\eqref{E:1normineq2p2} we find
\begin{align}
\label{E:thirdtermfirstfigp2}
    \includegraphics[scale=.28, valign = c]{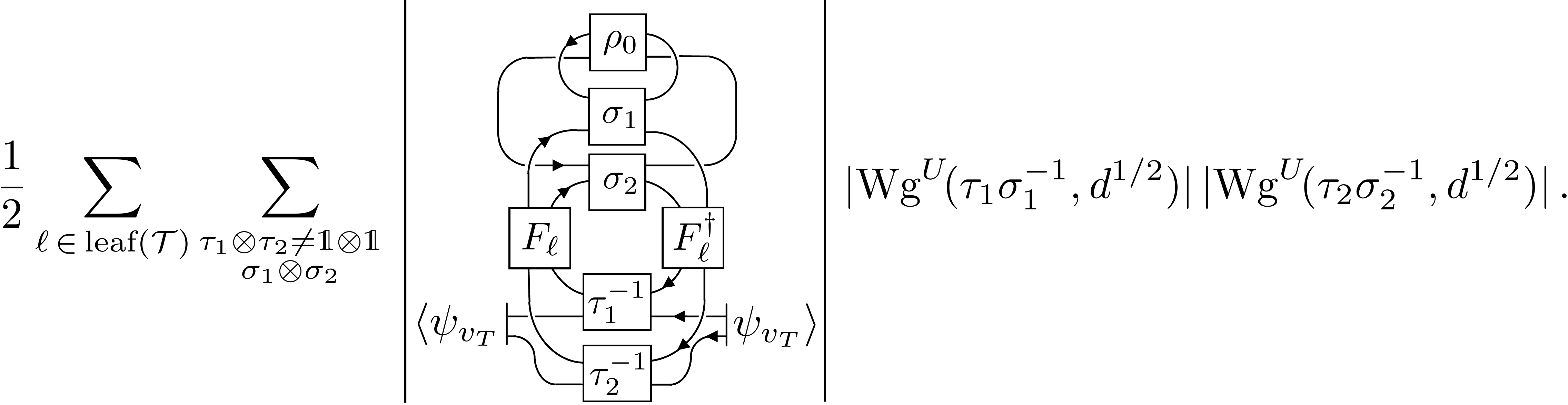}
\end{align}
If we label $\sigma = \sigma_1 \otimes \sigma_2$ and $\tau = \tau_1 \otimes \tau_2$, the proof proceeds identically to the third case of Proposition~\ref{Prop:prop1} up through~\eqref{E:ineqtofloor1}.  Then the new analog of~\eqref{E:usefullater1} is
\begin{equation}
\sum_{\ell \,\in\,\text{leaf}(\mathcal{T})} \sum_{\substack{\tau_1 \otimes \tau_2 \not = \mathds{1} \otimes \mathds{1} \\ \sigma_1, \sigma_2}} |p_{\sigma_1, \sigma_2 , \tau_1, \tau_2}(\ell)| \leq d^{T} \sum_{\sigma_1, \sigma_2} |\text{Wg}^U(\sigma_1^{-1}, d^{1/2})|\, |\text{Wg}^U(\sigma_2^{-1}, d^{1/2})| \sum_{\tau_1 \otimes \tau_2 \not = \mathds{1} \otimes \mathds{1}} d^{- \left\lfloor \frac{L(\tau^{-1})}{2} \right\rfloor}\,.
\end{equation}
The right-hand side can be bounded in part using $d^{T/2} \sum_{\sigma} |\text{Wg}^U(\sigma^{-1}, d^{1/2})|  \leq 1 + O(T^2/d^{1/2})$.  Since
\begin{equation}
\sum_{\substack{\tau_1 \otimes \tau_2 \not = \mathds{1} \otimes \mathds{1} \\ \tau_1, \tau_2 \,\in\, S_T}} d^{- \left\lfloor \frac{L(\tau^{-1})}{2}\right\rfloor}  \leq \sum_{\substack{\tau \not =  \mathds{1} \\ \tau\,\in\, S_T}} d^{- \left\lfloor \frac{L(\tau^{-1})}{2}\right\rfloor} \leq \frac{T^3}{d} + \frac{T^2}{d} + O\left(\frac{T^5}{d^2}\right)
\end{equation}
where the last bound comes from Proposition~\ref{Prop:prop1}, we find that if $T \leq o(d^{1/4})$ then
\begin{equation}
\frac{1}{2} \sum_{\ell \, \in \, \text{leaf}(\mathcal{T})} \sum_{\substack{\tau_1 \otimes \tau_2 \not = \mathds{1} \\ \sigma_1, \sigma_2}} |p_{\sigma_1, \sigma_2, \tau_1, \tau_2}(\ell)| \leq o(1)\,.
\end{equation}
$$$$
Putting all three cases together, we see that for $T \leq \Omega(d^{1/4})$ we have
\begin{equation}
\frac{1}{2}\sum_{\ell \,\in\, \text{\rm leaf}(\mathcal{T})} |p^{\mathcal{D}}(\ell) - \mathbb{E}_{\mathcal{U}_1, \mathcal{U}_2} [p^{\mathcal{U}_1 \otimes \mathcal{U}_2}(\ell)]| \leq \frac{1}{3}\,.
\end{equation}
\end{proof}

\bibliographystyle{unsrtnat}
\bibliography{refs}

\end{document}